\documentclass[11pt]{article}
\usepackage[margin=1in]{geometry}

\usepackage{graphicx}
\usepackage{indentfirst,csquotes}
\usepackage{algorithm}
\usepackage[noend]{algpseudocode}
\usepackage{mathtools}
\usepackage{tikz}
\usepackage{booktabs}
\usepackage[numbers]{natbib}
\usepackage{url}
\usepackage[
  colorlinks=true,
  linkcolor=black,
  citecolor=blue!60!black,
  urlcolor=blue!60!black
]{hyperref}
\usepackage{amsmath, amssymb, amsthm}
\usepackage{thmtools}
\usepackage{thm-restate}
\usepackage{enumitem}
\usepackage{xspace}
\usepackage{cleveref}

\newtheorem{theorem}{Theorem}[]

\newtheorem{observation}[theorem]{Observation}
\newtheorem{claim}[theorem]{Claim}
\newtheorem{example}[theorem]{Example}
\newtheorem{lemma}[theorem]{Lemma}
\newtheorem{proposition}[theorem]{Proposition}
\newtheorem{corollary}[theorem]{Corollary}

\newcommand{\LEFCAshort}{\textsc{LEFA}\xspace}
\newcommand{\LEFCASMshort}{\textsc{LWMEFA}\xspace}
\usepackage{authblk}

\providecommand{\shortcite}[1]{\citeyearpar{#1}}

\begin{document}
\title{Envy-Free Allocation of Indivisible Goods under Leontief Preferences}
\author[1]{Tanmay Inamdar$^{\text{1}}$\hspace{1.5cm}Pallavi Jain$^{\text{1}}$\hspace{1.5cm}Pranjal Pandey}
\affil[1]{\small Department of Computer Science and Engineering, Indian Institute of Technology Jodhpur, India.}
\date{}

\maketitle

\begin{abstract}
Envy-freeness is a fundamental notion of fairness in the allocation of indivisible goods. In this paper, we study envy-free allocation under Leontief preferences, which model perfect complements. Although Leontief preferences have been extensively studied in the context of allocating divisible goods and market equilibria, they have received comparatively little attention for the allocation of indivisible goods.
We show that, unlike additive valuations in cardinal preferences, an envy-free allocation always exists for Leontief preferences when there are at least two goods.  
In contrast, envy-free allocations may fail to exist when  there is a single good, however it can be decided in polynomial time. 
We next study the problem of computing a welfare-maximizing envy-free allocation. We prove that this problem is NP-hard in general,  
whereas it is polynomial-time solvable when there is only a single good or agents have identical demands. Finally, we investigate the parameterized complexity of this problem.
\end{abstract} 

\medskip

\textbf{AI Disclosure.} We used large language models (LLMs), in particular Claude Opus~5.5 (Anthropic), for language editing and for improving the presentation of proofs. LLM helped develop the initial ideas for the hardness reduction for $m$ (Theorem~\ref{thm:hard-m}) and for bounding the number of candidate witnesses (Lemma~\ref{lem:few-witnesses}), the authors simplified, wrote and verified the final proofs. All other results are the authors' own, and the authors take full responsibility for all results in this paper.

\bigskip

\noindent

\section{Introduction}

Consider a web platform that shares the capacity of its cluster, say $16$ GPUs, $64$ CPU cores and $16$ memory blocks of $32$\,GB, among the teams running its services. Resources are allocated in indivisible units, and every service scales by adding \emph{replicas}, each of which needs its resources in \emph{fixed proportions}; e.g., a replica of an LLM-based search assistant may need $4$ GPUs, $8$ cores and $4$ memory blocks. Extra GPUs without matching memory do not let a team serve more traffic, so the value of a bundle to a team is the number of complete replicas it can run, i.e., the minimum, over all required resources, of the number of units received divided by the number of units needed per replica. These are exactly \emph{Leontief preferences}, which model perfect complements~\cite{DBLP:conf/nsdi/GhodsiZHKSS10,10.1145/2739040}.

\begin{example}
\label{ex:intro}
Consider the cluster above and three teams with the following demands per replica.
\begin{center}
\small
\begin{tabular}{lccc}
\toprule
 & GPUs & cores & memory\\
\midrule
Team $A$ (search)         & $4$ & $8$  & $4$\\
Team $B$ (recommendation) & $2$ & $8$  & $2$\\
Team $C$ (indexing)       & $0$ & $16$ & $2$\\
\bottomrule
\end{tabular}
\end{center}
Suppose $A$, $B$ and $C$ run $2$, $3$ and $1$ replicas, i.e., receive $(8,16,8)$, $(6,24,6)$ and $(0,16,2)$, respectively. This allocation is envy-free, but $2$ GPUs and $8$ cores remain, which are useless to every team as all memory is used. Giving them to $B$ makes $C$ envy $B$, since $C$ can run $\min\{\lfloor 32/16\rfloor,\lfloor 6/2\rfloor\}=2$ replicas with $(8,32,6)$, whereas giving them to $A$ or $C$ keeps the allocation envy-free.
\end{example}

The example above is an instance of the fair allocation of indivisible resources, a fundamental problem in artificial intelligence, economics, and multi-agent systems, with Web applications such as cloud resource sharing, multi-tenant service hosting, and online advertising. It is also closely related to load balancing in web services: the replicas a service can run with its share of the cluster determine how much request traffic it can handle, so a fair split prevents some services from being overloaded while others hold resources they cannot use.

The existence and computational complexity of envy-free allocation of indivisible goods depend crucially on the underlying preference model and valuation function. Most of the algorithmic literature has focused on additive valuations and their generalizations, such as submodular and subadditive valuations. Under additive valuations, exact envy-free allocations may not exist, motivating approximate notions such as EF1 and EFX\cite{DBLP:journals/ai/AmanatidisABFLMVW23,DBLP:journals/sigecom/AzizLMW22,brandt2016handbook,DBLP:journals/jair/LiuLSW24,maskin1987fair,DBLP:conf/sigecom/LiptonMMS04,DBLP:conf/aaai/ChaudhuryGM21}. In contrast, much less is known about envy-freeness under complementary preferences, despite their importance in many real-world applications.

In this paper, we study envy-free allocation under Leontief preferences. As illustrated above, such preferences naturally arise whenever multiple resources are jointly required to complete a task, and additional copies of one resource are valuable only when accompanied by the other required resources. Although Leontief preferences have been extensively studied in market equilibrium and divisible resource allocation~\cite{DBLP:conf/infocom/0030LL14,DBLP:conf/aaai/NarayanaK21,DBLP:conf/nsdi/GhodsiZHKSS10,DBLP:journals/tcs/Garg17,DBLP:conf/stoc/GargMVY17,DBLP:journals/mp/ZhuDY12,DBLP:conf/icalp/CodenottiV04,DBLP:conf/wine/GoelHP19,li2013egalitarian}, their role in the fair allocation of indivisible goods has received little attention~\cite{10.1145/2739040}.

Throughout this paper, we require every copy of every good to be assigned, i.e., allocations are \emph{complete}. This is natural for several reasons. First, idle resources incur real costs, e.g., an idle GPU still consumes power, whereas resources that do not increase an agent's utility never hurt it and are often useful, e.g., spare cores speed up data loading and spare memory allows caching. Second, many allocation processes require every item to be owned for accountability or legal reasons, e.g., water rights follow the principle of ``use it or lose it''~\cite{lustgarten2015useit}. Third, a resource that is useless today may become valuable later. However, as Example~\ref{ex:intro} shows, completeness makes envy-freeness non-trivial: a resource that is useless to its owner may be exactly what another agent is missing, so an algorithm must decide not only how many tasks every agent receives, but also where all remaining copies go.

To the best of our knowledge, this is the first systematic study of envy-free complete allocation of indivisible goods under Leontief preferences. Our results reveal a striking contrast with additive valuations: \emph{while exact envy-free allocations may fail to exist under additive valuations, they are guaranteed under Leontief preferences when every agent demands at least two goods}. We further investigate the computational complexity of finding and optimising envy-free allocations. Our main contributions are summarized below.

\paragraph{Our Contribution.}
We study the following questions.
\begin{itemize}[leftmargin=*]
\item \textbf{Existence.} If every agent demands at least two goods, an envy-free allocation always exists and can be computed in $\mathcal{O}(mn)$ time (Thm.~\ref{thm:multi-good-ef}). This is tight: envy-free allocations may fail to exist if agents demand a single good, but if every agent demands at most one good, we can decide existence in polynomial time (Thm.~\ref{thm:single-good-types}).

  \item \textbf{Hardness of welfare maximization.} Envy-free allocations may have welfare $0$, so we study computing an envy-free allocation of maximum welfare. This problem is NP-hard even if every good has one copy and all demands are in $\{0,1\}$, W[1]-hard parameterized by the welfare, and inapproximable within $n^{1-\epsilon}$ (Thm.~\ref{thm:np-hard-welfare}); hence, it is para-NP-hard parameterized by $q_{\max}$ (Cor.~\ref{cor:para-np-hard}). Moreover, it is W[1]-hard parameterized by the number of goods $m$, even if all numbers are encoded in unary (Thm.~\ref{thm:hard-m}).

\item \textbf{Algorithms for welfare maximization.} The problem is polynomial time solvable if every agent demands at most one good (Thm.~\ref{thm:single-good-welfare}) or all agents have the same demands (Thm.~\ref{thm:identical-multi}). It is FPT parameterized by the number of agents $n$, with running time $2^{\mathcal{O}(n^2\log n)}\cdot\mathrm{poly}(n,m,\log V)$ (Thm.~\ref{thm:fpt-n}) as $i$ does not envy $k$ if the bundle of $k$ is short of \emph{some} good demanded by $i$, we guess such a witness good for each pair, and show that $n$ candidates per pair suffice instead of $m$ (Lemma~\ref{lem:few-witnesses}), after which an ILP with $n$ variables computes the utilities. If some good is demanded by all agents and every agent demands at least two goods, two carefully chosen agents can absorb all remaining copies without creating envy (Lemma~\ref{lem:completion}), so no guessing is needed, and the running time improves to $2^{\mathcal{O}(n\log n)}\cdot\mathrm{poly}(n,m,\log V)$, and to $2^{\mathcal{O}(n\log\log n)}\cdot\mathrm{poly}(n,m,\log V)$ with a randomized algorithm (Thm.~\ref{thm:large-common}). Complementing the hardness for $m$ and $q_{\max}$, it is FPT parameterized by the total number of copies $K$ (Thm.~\ref{thm:split-welfare}), and thus by $m+q_{\max}$ (Cor.~\ref{cor:split-welfare-mq}). If the agents have only $\tau$ distinct demand vectors, e.g., if many agents run replicas of the same service, the running time improves from $2^{\mathcal{O}(K\log K)}$ to $2^{\mathcal{O}(\tau+K)}$, which is faster whenever $\tau<K\log K$ (Thm.~\ref{thm:dp-welfare}).
  All these algorithms also decide whether an envy-free allocation exists.

  \item \textbf{EF1 and Pareto optimality.} A complete allocation that is EF1 and Pareto optimal always exists and can be computed in pseudo-polynomial time (Thm.~\ref{thm:ef1-po}), by completing the partial allocation of Parkes et al.~\shortcite{10.1145/2739040}, we show that some agent can receive all remaining copies without violating EF1, using a cycle argument based on Pareto optimality.
\end{itemize}

\paragraph{Related Work.}
Ghodsi et al.~\shortcite{DBLP:conf/nsdi/GhodsiZHKSS10} introduced Dominant Resource Fairness (DRF) for sharing divisible resources in computing clusters, which is EF, strategy-proof, and Pareto optimal under Leontief preferences; Narayana and Kash~\shortcite{DBLP:conf/aaai/NarayanaK21} extended these guarantees to agents with limited demands. Closest to our work, Parkes et al.~\shortcite{10.1145/2739040} considered indivisible tasks, showed that EF and Pareto optimality are incompatible, and gave an algorithm for a partial allocation that is EF1, Pareto optimal, and satisfies sharing incentives. We complete their allocation to a complete EF1 and Pareto optimal allocation (Thm.~\ref{thm:ef1-po}), and their incompatibility result motivates our study of welfare-maximizing envy-free allocations. Br\^{a}nzei et al.~\shortcite{DBLP:conf/ijcai/BranzeiLM16} studied single-minded agents, who have utility $1$ for a specific subset of goods and $0$ otherwise; this is the special case of $0/1$ demand vectors, which our reduction in Thm.~\ref{thm:np-hard-welfare} uses.

Barman et al.~\shortcite{DBLP:conf/atal/BarmanGJKN19} showed that maximizing the welfare among EF1 allocations is NP-hard under additive valuations, although an EF1 allocation always exists, and Bu et al.~\shortcite{DBLP:conf/wine/BuLLST23} showed strong inapproximability for maximizing an allocator's welfare subject to fairness constraints, by a reduction similar in spirit to ours from \textsc{Set Packing}. We observe a similar gap under Leontief preferences: welfare maximization is NP-hard (Thm.~\ref{thm:np-hard-welfare}) and W[1]-hard parameterized by the number of goods (Thm.~\ref{thm:hard-m}), even if every agent demands at least two goods, where an envy-free allocation always exists and can be found in polynomial time (Thm.~\ref{thm:multi-good-ef}).

Gafni et al.~\shortcite{DBLP:journals/teco/GafniHLT23} studied goods and chores in which every item comes in several copies, under additive valuations and with every agent receiving at most one copy of each item; in our model, an agent may receive any number of copies, which determines its utility through its demand vector. Bredereck et al.~\shortcite{DBLP:conf/ecai/BredereckKKN23} studied \emph{high-multiplicity} fair allocation with additive utilities and used integer programming to show that finding an envy-free Pareto-efficient allocation is FPT parameterized by the number of agents plus the number of item types, even for binary encoding; their multiplicities play the role of our copies $q_j$, and our FPT algorithm for $n$ also relies on integer programming with few variables. For background on parameterized complexity, see~\cite{Downey1999,CyganFKLMPPS2015}.

\section{Preliminaries}

\paragraph{Problem Instance.} An allocation instance under Leontief preferences is given by $(N,G,Q,D)$, where $N = \{1,2,\ldots,n\}$ denotes the set of agents, $G = \{g_1,\ldots,g_m\}$ denotes the set of $m$ \emph{types} of goods, $Q=(q_1,\ldots,q_m)$ is an $m$-dimensional vector, where $q_j > 0$ is the number of identical copies of good $g_j$, and $D=(d_{ij})_{n\times m}$ is a two-dimensional array of \emph{non negative integers}, where $d_{ij}$ denotes the demand of agent $i$ for good $g_j$. For brevity, we refer to a type of good simply as a \emph{good}; note that the total number of items is $K=\sum_{g_j\in G}q_j$. For an agent $i\in N$, we let $S_i=\{g_j\in G : d_{ij} > 0\}$ denote the set of goods \emph{demanded} by agent $i$. 
Throughout the paper, we assume that every agent demands at least one good, i.e., $S_i\neq\emptyset$ for every agent $i\in N$. Note that we allow $d_{ij}=0$, i.e., an agent need not demand every good. We let $q_{\max}=\max_{g_j\in G}q_j$ and $d_{\max}=\max_{i\in N,g_j\in G}d_{ij}$.

An \emph{allocation} is an ordered tuple of vectors
$\mathbf{X} = ( \mathbf{X}_1, \dots, \mathbf{X}_n ) \in (\mathbb{Z}_{\ge 0}^m)^n$,
where $\mathbf{X}_i = (x_{i1}, \dots, x_{im})$ denotes the \emph{bundle} assigned to the agent $i$. Here, $x_{ij}$ denotes the number of copies of good $g_j$ assigned to the agent $i$. The allocation must be complete, i.e., $\sum_{i \in N} x_{ij} = q_j$ for every good $g_j \in G$. In certain situations, we work with \emph{partial} allocations, i.e., allocations with $\sum_{i\in N}x_{ij} \le q_j$ for every good $g_j\in G$, where the inequality may be strict for some goods. However, unless explicitly mentioned otherwise, we always mean complete allocations.

\paragraph{Leontief Utility.} For a vector $\mathbf{y}=(y_1,\ldots,y_m)\in\mathbb{Z}_{\ge0}^m$, where $y_j$ is a number of copies of good $g_j$, the utility of agent $i$ is
$u_i(\mathbf{y}) \;=\; \min_{g_j \in S_i} \left\lfloor \frac{y_j}{d_{ij}} \right\rfloor.$
Intuitively, $u_i(\mathbf{y})$ is the maximum number of complete copies of its demand vector that agent $i$ can form from $\mathbf{y}$. The minimum is taken only over the goods demanded by agent $i$, since the goods with $d_{ij}=0$ do not affect its utility. In particular, the utility of agent $i$ for a bundle $\mathbf{X}_k=(x_{k1},\ldots,x_{km})$ of an allocation (where $k$ can be equal to $i$) is $u_i(\mathbf{X}_k)=\min_{g_j\in S_i}\lfloor x_{kj}/d_{ij}\rfloor$.

\paragraph{Fairness and Welfare.} An allocation $\mathbf{X}=(\mathbf{X}_1,\dots,\mathbf{X}_n)$ is said to be \emph{envy-free} (EF) if for every pair of agents $i,k \in N$, $u_i(\mathbf{X}_i) \ge u_i(\mathbf{X}_k)$. An allocation $\mathbf{X}$ is \emph{Pareto optimal} (PO) if there is no other allocation $\mathbf{X}'=(\mathbf{X}'_1,\dots,\mathbf{X}'_n)$ such that $u_i(\mathbf{X}'_i) \ge u_i(\mathbf{X}_i)$ for every agent $i \in N$ and $u_i(\mathbf{X}'_i) > u_i(\mathbf{X}_i)$ for at least one agent $i \in N$.

The \emph{welfare} of an allocation $\mathbf{X}$ is the total utility of all the agents, i.e., $\sum_{i\in N}u_i(\mathbf{X}_i)$. Clearly, if an allocation maximises the welfare, then it is also PO. Moreover, an allocation with maximum welfare amongst all envy-free allocations is also PO amongst all envy-free allocations.

Given an input $(N,G,Q,D)$, we call the problem of finding an envy-free allocation (if one exists) \emph{\LEFCAshort}, and the problem of finding a welfare-maximising envy-free allocation amongst all envy-free allocations (if one exists) \emph{\LEFCASMshort}.

\section{\LEFCAshort: Existence and  Computation}

We begin with a striking result: in contrast to additive utilities, under Leontief preferences, an envy-free allocation always exists and can be computed in polynomial time when every agent demands at least two goods positively. 

\begin{restatable}{theorem}{multi-good-ef}
\label{thm:multi-good-ef}
\LEFCAshort can be solved in $\mathcal{O}(mn)$ time if every agent demands at least two goods; in this case, a complete envy-free allocation always exists.
\end{restatable}

If there is only one agent, then giving all goods to this agent is trivially envy-free. Otherwise, $n\ge2$. Since an allocation in which every agent has utility $0$ for every bundle is envy-free, we construct such a complete allocation.

For every agent $i\in N$, fix a set $T_i\subseteq S_i$ of two goods, which exists since $|S_i|\ge2$. Let $H$ be the multigraph with a vertex for every good and an edge $e_i=T_i$ for every agent $i$; if $T_i=T_{i'}$ for distinct agents $i,i'$, then $H$ contains parallel edges. Thus, $H$ has exactly $n$ edges and no self-loops.

We now compute a proper colouring $c$ of $H$ with colours $\{1,\ldots,n\}$, i.e., the two goods of $T_i$ get different colours for every agent $i$. Let $\Delta$ be the maximum degree of $H$; since $H$ has $n$ edges, $\Delta\le n$. If $\Delta\le n-1$, the greedy colouring uses at most $\Delta+1\le n$ colours~\cite[Proposition~5.2.1]{Diestel17}. Otherwise, $\Delta=n$, so some vertex $v$ is an endpoint of all $n$ edges; we colour $v$ with colour $1$ and all other vertices with colour $2$, which is proper as every edge contains $v$.

We assign all $q_j$ copies of every good $g_j$ to agent $c(g_j)$; an agent whose colour is not used receives the empty bundle. This allocation $\mathbf{X}$ is complete. Consider two agents $i,k\in N$. The two goods of $T_i$ have different colours, so they are assigned to different agents, and $\mathbf{X}_k$ contains no copy of at least one of them, say $g_j$. As $g_j\in S_i$, we have $d_{ij}>0$, and thus $u_i(\mathbf{X}_k)\le\lfloor 0/d_{ij}\rfloor=0$.

Finally, computing the sets $S_i$ and $T_i$ takes $\mathcal{O}(mn)$ time, and constructing $H$, the colouring and the allocation takes $\mathcal{O}(m+n)$ time. Thus, the total running time is $\mathcal{O}(mn)$.

We next consider instances in which every agent demands at most one good, i.e., $|S_i|\le 1$ for every agent $i\in N$. 
In contrast to Theorem~\ref{thm:multi-good-ef}, an envy-free allocation may not exist in this case. The following example illustrates this with two agents.

\begin{example}
\label{example:ef-one-good}
Consider a single good with one copy and two agents, both with unit demand for the good.
In any complete allocation exactly one agent receives the copy of the good and the other agent envies.
\end{example}

However, existence of an EF allocation can be decided in this case, and if one exists, it can be computed in polynomial time, as we show next.

For a good $g_j\in G$, let $A_j=\{i\in N : S_i=\{g_j\}\}$ be the set of agents that demand only $g_j$. Since every agent demands at most one good, the sets $A_j$ are pairwise disjoint, and every agent not in $\bigcup_{j}A_j$ demands no good and has utility $0$ for every bundle. The utility of an agent $i\in A_j$ for a bundle $\mathbf{X}_k$ is $u_i(\mathbf{X}_k)=\lfloor x_{kj}/d_{ij}\rfloor$, i.e., agent $i$ only cares about the number of copies of $g_j$ in each bundle. Hence, envy-freeness decomposes over goods: an allocation is envy-free if and only if for every good $g_j$, agent $i\in A_j$, and agent $k\in N$, we have $\lfloor x_{ij}/d_{ij}\rfloor\ge\lfloor x_{kj}/d_{ij}\rfloor$. Consequently, the copies of a good $g_j$ with $A_j=\emptyset$ can be assigned arbitrarily.

Consider a good $g_j$ with $A_j\neq\emptyset$, and let $\mathbf{X}$ be an envy-free allocation. We observe that $x^\star_j=\max_{k\in N}x_{kj}$ completely determines the utility $t_{ij}=\lfloor x^\star_j/d_{ij}\rfloor$ of every agent $i\in A_j$, as no agent in $A_j$ can receive fewer than $t_{ij}d_{ij}$ copies of $g_j$ (otherwise, it would envy the agent holding the most copies of $g_j$). Note that the maximum is taken over all agents in $N$, since agents outside $A_j$ may also hold copies of $g_j$. For every good $g_j$ with $A_j\neq\emptyset$, let us define a function $f_j:\mathbb{N}\rightarrow\mathbb{N}$ such that
$f_j(x)=\sum_{i\in A_j}\left\lfloor \frac{x}{d_{ij}}\right\rfloor\cdot d_{ij}.$ Since $f_j$ is monotone and non-decreasing, the largest $x$ satisfying $f_j(x)\le q_j$ can be found by binary search.

{\bf Algorithm.} We start with the empty allocation $\mathbf{Y}$ and handle every good $g_j$ separately. If $A_j=\emptyset$, we give all $q_j$ copies of $g_j$ to agent~$1$. Otherwise, we compute by binary search the largest $x^\star_j\in[0,q_j]$ with $f_j(x^\star_j)\le q_j$, which will be the maximum number of copies of $g_j$ that an agent may hold. If $n\,x^\star_j<q_j$, the copies of $g_j$ cannot all be assigned in this way, and we return that no EF allocation exists. Otherwise, we give $\lfloor x^\star_j/d_{ij}\rfloor\cdot d_{ij}$ copies of $g_j$ to every agent $i\in A_j$, and assign the remaining $q_j-f_j(x^\star_j)$ copies of $g_j$ greedily such that every agent receives at most $x^\star_j$ copies of $g_j$. If no good leads to a rejection, we return $\mathbf{Y}$.

\begin{theorem}{single-good-types}
\label{thm:single-good-types} \footnote{The proof of Lemmas and Theorems marked with $\clubsuit$ can be found in appendix.}
\LEFCAshort can be solved in $\mathcal{O}(mn+n\log q_{\max})$ time if every agent demands at most one good.
\end{theorem}
Consider an \LEFCAshort instance in which $|S_i|\le 1$ for every agent $i\in N$. As argued above, an allocation is envy-free if and only if for every good $g_j$, every agent $i\in A_j$ and every agent $k\in N$, we have $\lfloor x_{ij}/d_{ij}\rfloor\ge\lfloor x_{kj}/d_{ij}\rfloor$; agents that demand no good never envy. Thus, it suffices to consider each good $g_j$ with $A_j\neq\emptyset$ separately. Copies of goods with $A_j=\emptyset$ are assigned to agent~$1$, which does not affect envy-freeness.

We first show that if the above algorithm returns no EF allocation exists, then no envy-free allocation exists. Suppose, for contradiction, that $\mathbf{X}$ is an envy-free allocation, and let $g_j$ be the good for which the algorithm returns no EF allocation exists, i.e., $n x^\star_j<q_j$. Let $p\in N$ be an agent receiving the maximum number of copies of $g_j$, i.e., $x_{pj}=\max_{k\in N}x_{kj}$; note that $p$ need not belong to $A_j$.  For every agent $i\in A_j$, we have $d_{ij}>0$ as $S_i=\{g_j\}$, and we let $t_{ij}=\lfloor x_{pj}/d_{ij}\rfloor$. If $i\neq p$, then envy-freeness requires $u_i(\mathbf{X}_i)\ge u_i(\mathbf{X}_p)$, i.e., $\lfloor x_{ij}/d_{ij}\rfloor\ge t_{ij}$; if $i=p$, this holds trivially. In both cases, $x_{ij}\ge t_{ij}d_{ij}$.

Since every agent receives a non-negative number of copies,
\[
  f_j(x_{pj})=\sum_{i\in A_j}t_{ij}d_{ij}\;\le\;\sum_{i\in A_j}x_{ij}\;\le\;\sum_{k\in N}x_{kj}\;=\;q_j .
\]
Moreover, since $x_{pj}=\max_{k\in N}x_{kj}$, every agent satisfies $x_{kj}\le x_{pj}$, and summing over all $n$ agents gives $q_j=\sum_{k\in N}x_{kj}\le n x_{pj}$. Now, let $x'=\min\{x_{pj},q_j\}$. Since $f_j$ is monotone non-decreasing, $f_j(x')\le f_j(x_{pj})\le q_j$, so $x'\in[0,q_j]$ is a feasible candidate in the binary search, and therefore $x'\le x^\star_j$. If $x'=x_{pj}$, then $n x^\star_j\ge n x_{pj}\ge q_j$; otherwise, $x'=q_j$ and $n x^\star_j\ge n q_j\ge q_j$. In both cases, $n x^\star_j\ge q_j$, which contradicts $n x^\star_j<q_j$.

Next, we show that if the algorithm returns an allocation $\mathbf{Y}$, then $\mathbf{Y}$ is complete and envy-free. Consider a good $g_j$ with $A_j\neq\emptyset$, for which $n x^\star_j\ge q_j$. For every agent $i\in A_j$, let $t_{ij}=\lfloor x^\star_j/d_{ij}\rfloor$. The algorithm first allocates $t_{ij}d_{ij}$ copies of $g_j$ to every agent $i\in A_j$ and no copy to other agents. By definition of $f_j$, this uses $f_j(x^\star_j)\le q_j$ copies, and leaves $R=q_j-f_j(x^\star_j)\ge 0$ copies unassigned. The assumption $n x^\star_j\ge q_j$ gives
\[
  R\;=\;q_j-f_j(x^\star_j)\;\le\;n x^\star_j-f_j(x^\star_j)
  \;=\;\sum_{i\in A_j}\bigl(x^\star_j-t_{ij}d_{ij}\bigr)+\sum_{k\in N\setminus A_j}x^\star_j,
\]
which is exactly the total slack available if every agent may receive up to $x^\star_j$ copies of $g_j$. 

We can therefore greedily assign the $R$ remaining copies such that every agent receives at most $x^\star_j$ copies of $g_j$, and thus all $q_j$ copies are assigned. Every good with $A_j=\emptyset$ is completely assigned to agent~$1$, so $\mathbf{Y}$ is complete.
To verify envy-freeness, fix a good $g_j$ with $A_j\neq\emptyset$, an agent $i\in A_j$ and an agent $\ell\neq i$. By construction, $t_{ij}d_{ij}\le y_{ij}\le x^\star_j<(t_{ij}+1)d_{ij}$, so $u_i(\mathbf{Y}_i)=\lfloor y_{ij}/d_{ij}\rfloor=t_{ij}$, and $y_{\ell j}\le x^\star_j$, so
\[
  u_i(\mathbf{Y}_\ell)=\left\lfloor\frac{y_{\ell j}}{d_{ij}}\right\rfloor
  \;\le\;\left\lfloor\frac{x^\star_j}{d_{ij}}\right\rfloor
  \;=\;t_{ij}\;=\;u_i(\mathbf{Y}_i).
\]
Since every agent that demands a good belongs to exactly one set $A_j$, and agents that demand no good never envy, no agent envies any other, and the allocation is envy-free.

Now to argue the running time, computing the sets $A_j$ takes $\mathcal{O}(mn)$ time. For a good $g_j$ with $A_j\neq\emptyset$, the binary search finds the largest $x$ satisfying $f_j(x)\le q_j$ in $\mathcal{O}(\log q_j)$ iterations. In each iteration, we need to evaluate $f_j(\cdot)$, which can be done in $\mathcal{O}(|A_j|)$ time. Since the sets $A_j$ are pairwise disjoint, all binary searches together take $\mathcal{O}(n\log q_{\max})$ time. For every good, the greedy allocation step takes $\mathcal{O}(n)$ time, i.e., $\mathcal{O}(mn)$ time in total. This gives a total running time of $\mathcal{O}(mn+n\log q_{\max})$.

A natural follow-up is to seek an EF allocation that also \emph{maximises} the total welfare, which we consider next.

\section{Welfare Maximising Envy-Free Allocation}
\label{sec:welfare}

A well-known challenge in fair allocation is that EF and PO cannot, in general, be achieved simultaneously for indivisible goods under Leontief utilities~\cite{10.1145/2739040}.

One might hope to recover compatibility by relaxing PO to
\emph{$\alpha$-Pareto optimality} ($\alpha$-PO), i.e., if we multiply utility of every agent by $(1+\alpha)$, where $\alpha\geq 0$, then it is PO.
However, the following example shows this is futile for any $\alpha \geq 0$. 

\begin{example}
Consider two agents and two goods, each with one copy ($q_1 = q_2 = 1$),
where both agents have identical demand $d_{i1} = d_{i2} = 1$.
Giving both goods to agent~$1$ yields utilities $(1, 0)$ and is PO, but agent~$2$ envies agent~$1$. Similarly, if we assign both goods to agent~$2$, then agent $1$ envies agent $2$.  

The only EF allocation in this instance is to split the goods so that each agent lacks at least one
demanded good (say give one good to agent $1$ and other to agent $2$). This allocation has utility $0$ for each agent, which is not $\alpha$-PO for any $\alpha \geq 0$.
\end{example}

Motivated by this fundamental incompatibility between EF and PO, we study the best achievable
goal within the EF constraint: computing an envy-free allocation that
\emph{maximises} the sum of utilities, i.e., \LEFCASMshort. Note that an EF allocation that maximises the sum of utilities is also PO within the space of EF allocations.

\subsection{Hardness}
\label{sub:hardness}
We show that the decision version of \LEFCASMshort (where we are given a target welfare $w$, and the task is to decide whether there exists a complete envy-free allocation with welfare at least $w$) is NP-complete, via a polynomial-time, reduction from \textsc{Set Packing}. In this problem, given a universe $U$, a family $\mathcal{S}=\{S_1,\ldots,S_r\}$ of subsets of $U$, and an integer $k\ge 0$, the question is whether $\mathcal{S}$ contains $k$ pairwise disjoint sets. This problem is known to be NP-complete~\cite{DBLP:conf/coco/Karp72} and is $W[1]$-hard when parameterised by the solution size $k$~\cite{CyganFKLMPPS2015}. Moreover, \textsc{Independent Set} on a graph $H=(V,E)$ is the special case of \textsc{Set Packing} in which $U=E$ and the family contains, for every vertex $v\in V$, the set of edges incident to $v$. Consequently, assuming the Exponential-Time Hypothesis~\cite{impagliazzo2001complexity}, there is no algorithm for \textsc{Set Packing} that runs in time $f(k)\cdot(|\mathcal{S}|+|U|)^{o(k)}$ for any computable function $f$~\cite{CyganFKLMPPS2015}, and the maximum size of a packing cannot be approximated in polynomial time within a factor of $|\mathcal{S}|^{1-\epsilon}$ for any fixed $\epsilon>0$ unless P${}={}$NP, since this holds already for \textsc{Independent Set}~\cite{Hastad1999,Zuckerman2007}. We prove the following theorem.
 
\begin{theorem}
\label{thm:np-hard-welfare}
\LEFCASMshort is NP-hard and W[1]-hard parameterized by the desired welfare $w$, even if every good has one copy and all demands are in $\{0,1\}$; unless ETH fails, it cannot be solved in $f(w)\cdot(n+m)^{o(w)}$ time for any computable function $f$. Moreover, unless P${}={}$NP, the maximum welfare cannot be approximated within a factor of $n^{1-\epsilon}$ in polynomial time for any fixed $\epsilon>0$.
\end{theorem}
 
\begin{proof}
Note that the NP-membership is trivial. Let $(U,\mathcal{S},k)$ be the given instance of \textsc{Set Packing}, where $\mathcal{S}=\{S_1,\ldots,S_r\}$. We construct an instance $(N,G,Q,D,w)$ of \LEFCASMshort as follows. For every set $S_i\in\mathcal{S}$ we create one agent $a_i$, giving $n=|\mathcal{S}|$. For every element $e\in U$ we create one \emph{element good} $g_e$, and for every agent $a_i$ we create one \emph{private good} $h_i$, giving $m=|U|+|\mathcal{S}|$. Every good has exactly one copy. The demands of agent $a_i$ are
\[
d_{a_i,g_e} \;=\; \begin{cases} 1, & e\in S_i,\\ 0, & e\notin S_i,\end{cases}
\qquad\qquad
d_{a_i,h_j} \;=\; \begin{cases} 1, & j=i,\\ 0, & j\ne i.\end{cases}
\]
That is, agent $a_i$ demands one copy of its private good and one copy of every element good corresponding to an element of $S_i$. The desired welfare is $w=k$. The instance is constructed in time polynomial in $|U|+|\mathcal{S}|$.

Note that every agent $a_i$ demands one copy of $h_i$, which has a single copy, so $u_{a_i}(\mathbf{X}_{a_\ell})\le x_{a_\ell,h_i}\le1$ for every allocation $\mathbf{X}$ and all agents $a_i,a_\ell$. Moreover, $u_{a_i}(\mathbf{X}_{a_\ell})=1$ if and only if $\mathbf{X}_{a_\ell}$ contains $h_i$ and all element goods $g_e$ with $e\in S_i$.

In the forward direction, suppose that $\mathcal{S}'\subseteq\mathcal{S}$ is a packing of size at least $k$. We construct an allocation $\mathbf{X}$ as follows. Every agent $a_i$ receives its private good $h_i$, and every agent $a_i$ with $S_i\in\mathcal{S}'$ additionally receives the element goods $g_e$ with $e\in S_i$. All remaining element goods are given to an arbitrary agent. Since the sets in $\mathcal{S}'$ are pairwise disjoint, no good is assigned twice, and since every good is assigned, $\mathbf{X}$ is complete. Every agent $a_i$ with $S_i\in\mathcal{S}'$ holds all goods it demands and thus has utility $1$, and all other agents have utility at least $0$; hence, the welfare of $\mathbf{X}$ is at least $|\mathcal{S}'|\ge k$. Note that some agents $a_j$ with $S_j\notin\mathcal{S}'$ may also have utility $1$, which only increases the welfare. Finally, $\mathbf{X}$ is envy-free: for distinct agents $a_i$ and $a_\ell$, the bundle $\mathbf{X}_{a_\ell}$ does not contain $h_i$, so $u_{a_i}(\mathbf{X}_{a_\ell})=0\le u_{a_i}(\mathbf{X}_{a_i})$.

In the reverse direction, suppose that $\mathbf{X}$ is a complete envy-free allocation with welfare at least $k$. Since every utility is in $\{0,1\}$, the set $T$ of agents with utility $1$ satisfies $|T|\ge k$. Every agent $a_i\in T$ holds all element goods $g_e$ with $e\in S_i$. Since every good has one copy, distinct agents $a_i,a_\ell\in T$ hold disjoint sets of goods, and hence $S_i\cap S_\ell=\emptyset$. Thus, $\mathcal{S}'=\{S_i : a_i\in T\}$ is a packing of size $|T|\ge k$.

It remains to argue the inapproximability. The forward direction maps a packing of size $p$ to a complete envy-free allocation of welfare at least $p$, and the reverse direction maps, in polynomial time, a complete envy-free allocation of welfare $W$ to a packing of size at least $W$. Hence, the maximum welfare equals the maximum size of a packing, and every $\alpha$-approximation for \LEFCASMshort yields an $\alpha$-approximation for \textsc{Set Packing}. Since $n=|\mathcal{S}|$ and \textsc{Set Packing} cannot be approximated within $|\mathcal{S}|^{1-\epsilon}$ in polynomial time for any fixed $\epsilon>0$ unless P${}={}$NP~\cite{Hastad1999,Zuckerman2007}, the same holds for \LEFCASMshort.
\end{proof}

Since the above proof has one copy of every good, we have the following result. 

\begin{corollary}
\label{cor:para-np-hard}
\LEFCASMshort is para-NP-hard parameterized by the maximum number $q_{\max}$ of copies of a good.
\end{corollary}
 
The reduction above uses many goods and demands in $\{0,1\}$. We next show that the problem remains hard even for few goods, by a reduction from \textsc{Multicolored Clique}: given a graph $H=(V,E)$ whose vertices are partitioned into colour classes $V_1,\ldots,V_k$, decide whether $H$ has a clique with exactly one vertex of every class. This problem is W[1]-hard parameterised by $k$~\cite{CyganFKLMPPS2015}.

\begin{theorem}{hard-m}
\label{thm:hard-m}
\LEFCASMshort is W[1]-hard parameterized by the number $m$ of goods, even if all numbers are encoded in unary.
\end{theorem}

Let $(H,V_1,\ldots,V_k)$ be the given instance of \textsc{Multicolored Clique}, where $k\ge2$, and let $E_{ij}$ be the set of edges between $V_i$ and $V_j$ for $i<j$. Let $L=|V|+1$ and $K=k+\binom{k}{2}$, and give every vertex $v$ a distinct label $\ell(v)\in[L-1]$.

We construct an instance $(N,G,Q,D,w)$ of \LEFCASMshort as follows. For every vertex $v$ we create one \emph{vertex agent} $a_v$, and for every edge $e$ we create one \emph{edge agent} $a_e$, giving $n=|V|+|E|$. For every colour $i$ we create one \emph{colour good} $g_i$ and for every $i<j$ one \emph{pair good} $h_{ij}$, each with one copy, and for every ordered pair $i\neq j$ we create two \emph{check goods} $c^+_{ij}$ and $c^-_{ij}$, each with $2L$ copies, giving $m=k+\binom{k}{2}+2k(k-1)$. The demands are as follows (all demands not mentioned are $0$).
\begin{itemize}[leftmargin=*]
  \item The vertex agent $a_v$ of a vertex $v\in V_i$ demands one copy of $g_i$, and $L+\ell(v)$ copies of $c^+_{ij}$ and $L-\ell(v)$ copies of $c^-_{ij}$ for every $j\neq i$.
  \item The edge agent $a_e$ of an edge $e=\{v,w\}\in E_{ij}$ with $v\in V_i$ and $w\in V_j$ demands one copy of $h_{ij}$, $L-\ell(v)$ copies of $c^+_{ij}$ and $L+\ell(v)$ copies of $c^-_{ij}$, and $L-\ell(w)$ copies of $c^+_{ji}$ and $L+\ell(w)$ copies of $c^-_{ji}$.
\end{itemize}
Let $\mathbf{d}_a$ denote the demand vector of agent $a$. That is, a vertex agent of colour $i$ needs the colour good $g_i$, an edge agent between $V_i$ and $V_j$ needs the pair good $h_{ij}$, and the check goods $c^\pm_{ij}$ compare the label of a vertex of $V_i$ with the label of the endpoint in $V_i$ of an edge between $V_i$ and $V_j$. The desired welfare is $w=K$. All numbers are at most $2L$, so the instance is constructed in time polynomial in $|V|$, even in unary.

In the forward direction, suppose that $v_1,\ldots,v_k$ with $v_i\in V_i$ form a clique. Let $S$ be the set of the $K$ agents corresponding to this clique, i.e., the vertex agents $a_{v_1},\ldots,a_{v_k}$ and the edge agents $a_{\{v_i,v_j\}}$ for $i<j$. We construct an allocation $\mathbf{X}$ as follows. Every agent $a\in S$ receives $\mathbf{d}_a$, and all other agents receive nothing. Every colour good $g_i$ is received only by $a_{v_i}$, every pair good $h_{ij}$ only by $a_{\{v_i,v_j\}}$, and the check goods $c^\pm_{ij}$ only by $a_{v_i}$ and $a_{\{v_i,v_j\}}$, who receive $(L+\ell(v_i))+(L-\ell(v_i))=2L$ copies of $c^+_{ij}$ and, similarly, $2L$ copies of $c^-_{ij}$. Thus, $\mathbf{X}$ is complete, and as every agent in $S$ has utility $1$, its welfare is $K$.

Next, we show that $\mathbf{X}$ is envy-free. No agent can have utility more than $1$, as every agent demands one copy of a colour or pair good, which has a single copy. Hence, the agents in $S$ do not envy anyone, and no agent envies the empty bundles of the agents outside $S$. It remains to show that an agent $a'\notin S$ does not envy an agent $a\in S$. First, let $a=a_v$ for $v\in V_i$. If $a'$ is an edge agent or a vertex agent of a colour $j\neq i$, then $a'$ does not envy $a$, as $a$ received no pair good and no copy of $g_j$. Otherwise, $a'=a_u$ for some $u\in V_i\setminus\{v\}$, so $\ell(u)\neq\ell(v)$. For any $j\neq i$, agent $a'$ demands $L+\ell(u)$ copies of $c^+_{ij}$ and $L-\ell(u)$ copies of $c^-_{ij}$, whereas $a$ received $L+\ell(v)$ and $L-\ell(v)$ copies, respectively. Hence, $a$'s bundle contains too few copies of $c^+_{ij}$ for $a'$ if $\ell(v)<\ell(u)$, and of $c^-_{ij}$ if $\ell(v)>\ell(u)$. Now, let $a=a_e$ for an edge $e=\{u,v\}\in E_{ij}$ with $u\in V_i$ and $v\in V_j$. If $a'$ is a vertex agent or an edge agent not in $E_{ij}$, then $a'$ does not envy $a$, as $a$ received no colour good and not the pair good demanded by $a'$. Otherwise, $a'=a_{e'}$ for an edge $e'=\{u',v'\}\in E_{ij}\setminus\{e\}$ with $u'\in V_i$ and $v'\in V_j$, so $u'\neq u$ or $v'\neq v$; by symmetry between $i$ and $j$, assume $u'\neq u$, so $\ell(u')\neq\ell(u)$. Agent $a'$ demands $L-\ell(u')$ copies of $c^+_{ij}$ and $L+\ell(u')$ copies of $c^-_{ij}$, whereas $a$ received $L-\ell(u)$ and $L+\ell(u)$ copies, respectively. Hence, $a$'s bundle contains too few copies of $c^-_{ij}$ for $a'$ if $\ell(u)<\ell(u')$, and of $c^+_{ij}$ if $\ell(u)>\ell(u')$. Thus, $\mathbf{X}$ is envy-free.

In the reverse direction, suppose that $\mathbf{X}$ is a complete envy-free allocation with welfare at least $K$. We call an agent \emph{selected} if it has utility $1$, i.e., it holds at least its demand vector. As above, every agent has utility at most $1$. Since $g_i$ and $h_{ij}$ have a single copy, at most one vertex agent of $V_i$ is selected for every $i$, and at most one edge agent of $E_{ij}$ for every $i<j$. We now show that a selected edge agent leaves no room for vertex agents other than its endpoints. Let $a_e$ be a selected edge agent for $e=\{u,w\}\in E_{ij}$ with $u\in V_i$ and $w\in V_j$. As $a_e$ holds at least $L-\ell(u)$ copies of $c^+_{ij}$ and $L+\ell(u)$ copies of $c^-_{ij}$, at most $L+\ell(u)$ and $L-\ell(u)$ copies, respectively, remain for the other agents. A selected vertex agent $a_x$ with $x\in V_i$ needs $L+\ell(x)$ copies of $c^+_{ij}$ and $L-\ell(x)$ copies of $c^-_{ij}$, so $\ell(x)\le\ell(u)$ and $\ell(x)\ge\ell(u)$, i.e., $x=u$. Hence, $a_u$ is the only vertex agent of $V_i$ that can be selected, and symmetrically, using the goods $c^\pm_{ji}$, $a_w$ is the only vertex agent of $V_j$ that can be selected. Since at most one agent is selected for each of the $k$ colours and each of the $\binom{k}{2}$ pairs, and the welfare is at least $K=k+\binom{k}{2}$, exactly one vertex agent $a_{v_i}$ with $v_i\in V_i$ is selected for every $i$, and exactly one edge agent $a_{e_{ij}}$ with $e_{ij}\in E_{ij}$ is selected for every $i<j$. By the above, the selected vertex agents of $V_i$ and $V_j$ are the agents of the endpoints of $e_{ij}$, i.e., $e_{ij}=\{v_i,v_j\}\in E$ for all $i<j$. Hence, $v_1,\ldots,v_k$ are pairwise adjacent and form a multicoloured clique. Since $m=\mathcal{O}(k^2)$ and all numbers are at most $2L$, this proves W[1]-hardness parameterised by $m$, even for unary numbers.

Since \LEFCASMshort is NP-complete in general, we next investigate tractable cases: polynomial-time algorithms for restricted instances, and parameterized algorithms.

\subsection{Polynomial-Time Solvable Cases}

We first consider the case where every agent demands at most one good and show that if an envy-free allocation exists, then the allocation returned by the algorithm of Theorem~\ref{thm:single-good-types} maximizes the welfare.

\begin{theorem}{single-good-welfare}
\label{thm:single-good-welfare} 
\LEFCASMshort can be solved in $\mathcal{O}(mn+n\log q_{\max})$ time if every agent demands at most one good.
\end{theorem}

Let $(N,G,Q,D)$ be an instance of \LEFCASMshort in which every agent demands at most one good, and suppose that an EF allocation exists for this instance. The allocation $\mathbf{X}$ returned by Algorithm in Theorem~\ref{thm:single-good-types} is complete and envy-free. Moreover, for every good $g_j$ with $A_j\neq\emptyset$ and every agent $i\in A_j$, the utility of agent $i$ in $\mathbf{X}$ is $u_i(\mathbf{X}_i)=\lfloor x^\star_j/d_{ij}\rfloor$, and every agent that demands no good has utility $0$. Since the sets $A_j$ are pairwise disjoint, the welfare of $\mathbf{X}$ is $\sum_{g_j\in G:\,A_j\neq\emptyset}\ \sum_{i\in A_j}\left\lfloor \frac{x^\star_j}{d_{ij}}\right\rfloor.$

It remains to show that no complete envy-free allocation achieves strictly higher welfare. Let $\mathbf{Y}$ be a complete envy-free allocation, consider a good $g_j$ with $A_j\neq\emptyset$, and let $y^\star_j=\max_{k\in N}y_{kj}$. As argued in the proof of Theorem~\ref{thm:single-good-types} (with $x_{pj}=y^\star_j$), envy-freeness gives $f_j(y^\star_j)\le q_j$, and since $y^\star_j\le q_j$, the choice of $x^\star_j$ gives $y^\star_j\le x^\star_j$. 
Since every agent $i\in A_j$ demands only $g_j$ and $y_{ij}\le y^\star_j$, we have $u_i(\mathbf{Y}_i)=\lfloor y_{ij}/d_{ij}\rfloor\le\lfloor y^\star_j/d_{ij}\rfloor\le\lfloor x^\star_j/d_{ij}\rfloor=u_i(\mathbf{X}_i)$. Every agent that demands no good has utility $0$ in both $\mathbf{X}$ and $\mathbf{Y}$. Summing over all agents, the welfare of $\mathbf{Y}$ is at most the welfare of $\mathbf{X}$, so $\mathbf{X}$ maximizes the welfare.

If no EF allocation exists, the algorithm correctly reports this, and the running time follows, by Theorem~\ref{thm:single-good-types}.

Next, we consider identical demands where all agents have same demands. The allgorithm is based on idea that all agents must have the same utility in every envy-free allocation, so we give every agent as many tasks as possible and distribute the remaining copies without increasing any utility.

\begin{theorem}{identical-multi}
\label{thm:identical-multi}
\LEFCASMshort can be solved in $\mathcal{O}(mn+n\log q_{\max})$ time if all agents have the same demand vector.
\end{theorem}
\begin{proof}
Let $\mathbf{d}$ be the common demand vector, and let $S=\{g_j\in G: d_j>0\}$. If $|S|=1$, then every agent demands at most one good, and the claim follows from Theorem~\ref{thm:single-good-welfare}. If $n=1$, giving all copies to the single agent is optimal. Hence, assume that $|S|\ge2$ and $n\ge2$. Since all agents have the same utility function $u$, they have the same utility in every envy-free allocation, as otherwise an agent with lower utility envies one with higher utility. Let $t^\star=\min_{g_j\in S}\lfloor q_j/(n\,d_j)\rfloor$. We first construct an allocation $\mathbf{X}$. Every agent receives $t^\star d_j$ copies of every good $g_j\in S$, which is possible as $n\,t^\star d_j\le q_j$. Then, for a fixed good $g_a\in S$, we give all remaining copies of $g_a$ to agent $1$, and all remaining copies of all other goods to agent $2$. Thus, $\mathbf{X}$ is complete.

Next, we show that $\mathbf{X}$ is envy-free with welfare $n\,t^\star$. After the first step, every bundle consists of exactly $t^\star$ tasks. Afterwards, agent $1$ receives additional copies only of $g_a$, so its bundle still contains exactly $t^\star d_j$ copies of some $g_j\in S\setminus\{g_a\}$, which exists as $|S|\ge2$. Agent $2$ receives no additional copies of $g_a$, so its bundle still contains exactly $t^\star d_a$ copies of $g_a$. Hence, every bundle, including those of agents $1$ and $2$, still contains only $t^\star$ tasks, so every agent has utility exactly $t^\star$ for every bundle. Thus, no agent envies another, and the welfare is $n\,t^\star$. It remains to show that $\mathbf{X}$ has maximum welfare. In every complete envy-free allocation, all agents have the same utility $t$, so every agent holds at least $t\,d_j$ copies of every $g_j\in S$. This gives $n\,t\,d_j\le q_j$ for every $g_j\in S$, i.e., $t\le t^\star$, so the welfare is at most $n\,t^\star$. Finally, computing $t^\star$ takes $\mathcal{O}(m)$ time and constructing $\mathbf{X}$ takes $\mathcal{O}(mn)$ time, so this case takes $\mathcal{O}(mn)$ time.
\end{proof}

\subsection{FPT wrt number of agents}
\label{sub:fpt-n}

We now show that \LEFCASMshort is FPT when parameterised by the number of agents $n$. Our algorithm has three steps. First, for every ordered pair of agents $(i,k)$, it guesses a good of which the bundle of $k$ will contain appropriately few copies for agent $i$, so that $i$ does not envy $k$. Second, for each such guess, it computes the best utilities of the agents by an ILP with $n$ variables. Third, it distributes the copies of every good greedily, giving every agent enough copies for its utility, but not so many that another agent starts to envy it. We now introduce the notation used in these steps.

A \emph{pair} $i\to k$ consists of two distinct agents $i$ and $k$; we use it to look at the bundle of $k$ from the perspective of $i$. Let $\mathbf{t}=(t_1,\ldots,t_n)\in\mathbb{Z}_{\ge0}^n$ be a vector of \emph{target utilities}, where $t_i$ is the utility we aim for agent $i$; we do not know in advance whether $\mathbf{t}$ can be achieved. For every agent $k$ and good $g_j$, let $L_{kj}=t_k\,d_{kj}$ be the number of copies of $g_j$ that agent $k$ needs for utility $t_k$. If a bundle contains at most $(t_i+1)d_{ij}-1$ copies of some good $g_j\in S_i$, then agent $i$ has utility at most $t_i$ for it. We call such a good a \emph{witness} for $ik$ if the bundle of $k$ contains at most $(t_i+1)d_{ij}-1$ copies of it; then $i$ does not envy $k$. A \emph{witness map} $w$ maps every pair $i\to k$ to a good $w_{i k}\in S_i$. For every agent $k$ and good $g_j$, let $P_{kj}=\{i : i\to k \text{ is a pair with } w_{ik}=g_j\}$ be the agents whose witness for $k$ is $g_j$, and let
$U_{kj}=\min_{i\in P_{kj}}\bigl((t_i+1)\,d_{ij}-1\bigr),$
where $U_{kj}=\infty$ if $P_{kj}=\emptyset$. Thus, the bundle of $k$ may contain at most $U_{kj}$ copies of $g_j$, as otherwise some agent $i\in P_{kj}$ could obtain utility more than $t_i$ from it. We say that $(\mathbf{t},w)$ is \emph{feasible} if
\begin{enumerate}
  \item[(F1)] $L_{kj}\le(t_i+1)\,d_{ij}-1$ for every pair $i\to k$ and $g_j=w_{ik}$,
  \item[(F2)] $\sum_{k\in N}L_{kj}\le q_j$ for every good $g_j\in G$, and
  \item[(F3)] $\sum_{k\in N}U_{kj}\ge q_j$ for every good $g_j\in G$.
\end{enumerate}
Condition (F1) states that agent $k$ can receive the $L_{kj}$ copies it needs while $g_j=w_{ik}$ remains a witness for $i\to k$, i.e., $L_{kj}\le U_{kj}$ for every agent $k$ and good $g_j$. Conditions (F2) and (F3) state that the $q_j$ copies of $g_j$ can be distributed within these bounds; (F3) holds trivially if $U_{kj}=\infty$ for some agent $k$. We will show in Lemma~\ref{lem:witness-char} that feasible pairs correspond exactly to complete envy-free allocations.

\paragraph{Step 1. Guessing the witnesses.}
A witness for $i\to k$ can be any goods from $S_i$, and $|S_i|$ could be as large as $m$. However, we will show that it suffices to consider $n$ candidates. For a pair $i\to k$, order the goods of $S_i$ by non-decreasing ratio $d_{kj}/d_{ij}$, breaking ties by index, and let $\Gamma_{ik}$ be the set of the first $\min\{n,|S_i|\}$ goods in this order. Since (F1) holds for $w_{ik}=g_j$ if and only if $t_k\,d_{kj}/d_{ij}<t_i+1$, the goods that satisfy (F1) form a prefix of this order, so replacing a witness by an earlier good preserves (F1). Such a replacement can only violate (F3) for the new good, and only if this good already bounds the bundle of every agent other than $k$.Every agent $k''$ appears in only $n-1$ pairs $i''\to k''$, so at most $n-1$ goods can bound its bundle, and one of the first $n$ goods is always safe.

\begin{lemma}{few-witnesses}
\label{lem:few-witnesses}
If $(\mathbf{t},w)$ is feasible, then there is a witness map $w'$ with $w'_{ik}\in\Gamma_{ik}$ for every pair $i\to k$ such that $(\mathbf{t},w')$ is feasible.
\end{lemma}

    Suppose that $w_{ik}\notin\Gamma_{ik}$ for some pair $i\to k$. Then $|\Gamma_{ik}|=n$, and as the goods $g_j$ with $t_kd_{kj}\le(t_i+1)d_{ij}-1$ form a prefix of the order that contains $w_{ik}$ by (F1), every good of $\Gamma_{ik}$ satisfies this inequality. Fix an agent $k''\neq k$. At most $n-1$ goods $g_j$ have $U_{k''j}<\infty$, namely the witnesses of the pairs $i''\to k''$, so some $g_{j'}\in\Gamma_{ik}$ has $U_{k''j'}=\infty$. Let $w'$ agree with $w$ except that $w'_{ik}=g_{j'}$. Then (F1) holds as $g_{j'}\in\Gamma_{ik}$, (F2) is unchanged, and for (F3), the bound $U_{k,w_{ik}}$ can only increase, and (F3) holds for $g_{j'}$ as $U_{k''j'}=\infty$. Thus, $(\mathbf{t},w')$ is feasible, and repeating this for every pair yields the claim.

By Lemma~\ref{lem:few-witnesses}, it suffices to consider witness maps $w$ with $w_{ik}\in\Gamma_{ik}$ for every pair $i\to k$; we call them \emph{candidate} witness maps. There are at most $n^{n(n-1)}$ of them, and we \emph{guess} (i.e., enumerate) all of them one by one. In Step~2, we need to express the minimum in $U_{kj}$ by linear constraints. Therefore, for every candidate witness map $w$, we also guess, for every agent $k$ and good $g_j$ with $P_{kj}\neq\emptyset$, the agent $\sigma_{kj}\in P_{kj}$ that attains this minimum, i.e., whose bound on the bundle of $k$ is the tightest. We call $\sigma$ a \emph{tightness map} for $w$. For every agent $k$, the sets $P_{kj}$ are disjoint and contain $n-1$ agents in total, so there are at most $(n-1)^{n-1}$ choices for every agent $k$ and at most $n^{n(n-1)}$ tightness maps in total. For each guess $(w,\sigma)$, we perform the next two steps. For some of the guesses, we can find a complete envy-free allocation at the end of Steps 2 and 3. Amongst all such allocations, we return an allocation with maximum welfare. If, for none of the guesses, we can find a complete envy-free allocation, then we output that no envy-free allocation exists.

\paragraph{Step 2: Computing the utilities by solving an ILP}
For a fixed guess $(w,\sigma)$, we want to compute target utilities $\mathbf{t}$ with maximum $\sum_it_i$ such that $(\mathbf{t},w)$ is feasible. Conditions (F1) and (F2) are linear in $\mathbf{t}$, and the tightness map $\sigma$ allows us to express (F3) linearly. Let $G^\star$ be the set of goods $g_j$ with $P_{kj}\neq\emptyset$ for every agent $k$, and consider the following ILP, denoted by $\mathrm{ILP}(w,\sigma)$, with an integer variable $t_i$ for every agent $i$.
\begin{alignat}{2}
  \max\quad & \textstyle\sum_{i\in N}t_i && \label{w:obj}\\
  \text{s.t.}\quad
  & t_k\,d_{kj}\le(t_i+1)\,d_{ij}-1
    &\ & \forall\, i\to k,\,g_j=w_{ik}, \label{w:f1}\\
  & \textstyle\sum_{k\in N}t_k\,d_{kj}\le q_j
    && \forall\, g_j\in G, \label{w:f2}\\
  & (t_{\sigma_{kj}}+1)\,d_{\sigma_{kj}j}\le(t_i+1)\,d_{ij}
    && \forall\, k,j,\ i\in P_{kj}, \label{w:tight}\\
  & \textstyle\sum_{k\in N}\bigl((t_{\sigma_{kj}}+1)\,d_{\sigma_{kj}j}-1\bigr)\ge q_j
    && \forall\, g_j\in G^\star, \label{w:f3}\\
  & t_i\in\mathbb{Z}_{\ge0}
    && \forall\, i\in N. \label{w:int}
\end{alignat}
The objective~\eqref{w:obj} is the welfare. Constraint~\eqref{w:f1} is (F1): agent $k$ needs $t_kd_{kj}$ copies of its witness $g_j=w_{ik}$, and these must be at most $(t_i+1)d_{ij}-1$, so that $i$ has utility at most $t_i$ for the bundle of $k$. Constraint~\eqref{w:f2} is (F2): there are enough copies of every good to give every agent its target utility. Constraint~\eqref{w:tight} ensures that $\sigma_{kj}$ indeed attains the minimum in $U_{kj}$, i.e., $U_{kj}=(t_{\sigma_{kj}}+1)d_{\sigma_{kj}j}-1$. Given this, \eqref{w:f3} is (F3): all copies of $g_j$ fit below the upper bounds; for a good $g_j\notin G^\star$, some $U_{kj}=\infty$, and (F3) holds trivially. We solve $\mathrm{ILP}(w,\sigma)$ using the following result.

\begin{proposition}[\cite{Lenstra1983,Kannan1987,ReisR23}]
\label{lem:lenstra}
An MILP with $p$ integer variables and input size $L$ can be solved in $p^{\mathcal{O}(p)}\cdot\mathrm{poly}(L)$ time. If all variables are integer, it can be solved by a randomized algorithm in $(\log p)^{\mathcal{O}(p)}\cdot\mathrm{poly}(L)$ time that returns the correct answer with high probability.
\end{proposition}

\begin{lemma}{witness-ilp}
\label{lem:witness-ilp}
Let $w$ be a witness map and let $\mathbf{t}\in\mathbb{Z}_{\ge0}^n$. Then $(\mathbf{t},w)$ is feasible if and only if $\mathbf{t}$ is a solution of $\mathrm{ILP}(w,\sigma)$ for some tightness map $\sigma$ for $w$. Moreover, $\mathrm{ILP}(w,\sigma)$ can be solved in $n^{\mathcal{O}(n)}\cdot\mathrm{poly}(n,m,\log V)$ time, where $V=\max\{q_{\max},d_{\max}\}$.
\end{lemma}

Constraints~\eqref{w:f1} and~\eqref{w:f2} are exactly (F1) and (F2), so it suffices to consider (F3).

In the forward direction, let $(\mathbf{t},w)$ be feasible, and for every agent $k$ and good $g_j$ with $P_{kj}\neq\emptyset$, let $\sigma_{kj}\in P_{kj}$ be an agent attaining the minimum in $U_{kj}$. Then~\eqref{w:tight} holds by the choice of $\sigma_{kj}$, and for every good $g_j\in G^\star$, the left-hand side of~\eqref{w:f3} is $\sum_{k\in N}U_{kj}\ge q_j$ by (F3).

In the reverse direction, let $\mathbf{t}$ be a solution of $\mathrm{ILP}(w,\sigma)$. By~\eqref{w:tight}, $\sigma_{kj}$ attains the minimum in $U_{kj}$, so $U_{kj}=(t_{\sigma_{kj}}+1)d_{\sigma_{kj}j}-1$ whenever $P_{kj}\neq\emptyset$. Hence, for every good $g_j\in G^\star$, \eqref{w:f3} gives $\sum_{k\in N}U_{kj}\ge q_j$, and for every other good $g_j$, we have $U_{kj}=\infty$ for some agent $k$. Thus, (F3) holds.

For the running time, $\mathrm{ILP}(w,\sigma)$ is bounded by~\eqref{w:f2} and has $n$ integer variables and at most $n(n-1)$ constraints of each of the types~\eqref{w:f1} and~\eqref{w:tight}, and at most $m$ constraints of each of the types~\eqref{w:f2} and~\eqref{w:f3}. All entries are integers of absolute value at most $nV$, so the time bound follows from Proposition~\ref{lem:lenstra}.

\paragraph{Step 3: Constructing the allocation.}
Given a feasible pair $(\mathbf{t},w)$, its \emph{greedy allocation} is obtained as follows: for every good $g_j$, give $L_{kj}$ copies to every agent $k$, and distribute the remaining copies arbitrarily such that every agent $k$ receives at most $U_{kj}$ copies in total. The following lemma shows that this is always possible and yields a complete envy-free allocation, and that, conversely, every complete envy-free allocation gives a feasible pair.

\begin{lemma}{witness-char}
\label{lem:witness-char}
There is a complete envy-free allocation with welfare at least $W\in\mathbb{Z}_{\ge0}$ if and only if there is a feasible pair $(\mathbf{t},w)$ with $\sum_{i\in N}t_i\ge W$. Moreover, given such a pair, its greedy allocation is such an allocation and can be computed in polynomial time.
\end{lemma}

In the forward direction, let $\mathbf{X}$ be a complete envy-free allocation with welfare at least $W$, let $t_i=u_i(\mathbf{X}_i)$ for every $i\in N$, and for every pair $i\to k$, let $w_{ik}\in S_i$ be a good attaining the minimum in $u_i(\mathbf{X}_k)=\min_{g_j\in S_i}\lfloor x_{kj}/d_{ij}\rfloor$. In this direction, the bounds $L_{kj}$ and $U_{kj}$ are defined with respect to these utilities $\mathbf{t}$ of $\mathbf{X}$ and this witness map $w$. As $i$ does not envy $k$, we get $x_{k,w_{ik}}\le(t_i+1)d_{i,w_{ik}}-1$, so $x_{kj}\le U_{kj}$, and as $u_k(\mathbf{X}_k)=t_k$, we get $x_{kj}\ge L_{kj}$ for every agent $k$ and good $g_j$. Thus, $L_{kj}\le x_{kj}\le U_{kj}$, which gives (F1), and summing over all agents using $\sum_kx_{kj}=q_j$ gives (F2) and (F3).

In the reverse direction, let $\mathbf{X}$ be the greedy allocation of a feasible pair $(\mathbf{t},w)$. By (F1)--(F3), $L_{kj}\le U_{kj}$ and $\sum_kL_{kj}\le q_j\le\sum_kU_{kj}$ for every good $g_j$, so the greedy step assigns all copies with $L_{kj}\le x_{kj}\le U_{kj}$, and $\mathbf{X}$ is complete. As $x_{ij}\ge t_id_{ij}$, we have $u_i(\mathbf{X}_i)\ge t_i$, so the welfare of $\mathbf{X}$ is at least $\sum_it_i$. For a pair $i\to k$ and $g_j=w_{ik}$, we have $x_{kj}\le U_{kj}\le(t_i+1)d_{ij}-1$, so $u_i(\mathbf{X}_k)\le t_i\le u_i(\mathbf{X}_i)$, and $\mathbf{X}$ is envy-free. Finally, computing all bounds $L_{kj}$ and $U_{kj}$ takes $\mathcal{O}(nm+n^2)$ time, and the greedy allocation then distributes the copies of every good in $\mathcal{O}(n)$ time, so the greedy allocation can be computed in $\mathcal{O}(nm+n^2)$ time.

\begin{theorem}{fpt-n}
\label{thm:fpt-n}  
\LEFCASMshort can be solved in $2^{\mathcal{O}(n^2\log n)}\cdot\mathrm{poly}(n,m,\log V)$ time, where $V=\max\{q_{\max},d_{\max}\}$; i.e., it is FPT parameterized by the number $n$ of agents.
\end{theorem}

By Lemmas~\ref{lem:witness-ilp} and~\ref{lem:witness-char}, every solution $\mathbf{t}$ of some $\mathrm{ILP}(w,\sigma)$ gives a feasible pair $(\mathbf{t},w)$, whose greedy allocation is complete, envy-free, and has welfare at least $\sum_it_i$. Hence, if no complete envy-free allocation exists, then every $\mathrm{ILP}(w,\sigma)$ is infeasible, and the algorithm correctly reports this. Otherwise, let $W^\star$ be the maximum welfare of a complete envy-free allocation. By Lemma~\ref{lem:witness-char}, there is a feasible pair $(\mathbf{t},w)$ with $\sum_it_i\ge W^\star$; by Lemma~\ref{lem:few-witnesses}, we may assume that $w$ is a candidate witness map; and by Lemma~\ref{lem:witness-ilp}, $\mathbf{t}$ is a solution of $\mathrm{ILP}(w,\sigma)$ for some tightness map $\sigma$. Hence, for this guess $(w,\sigma)$, an optimal solution $\mathbf{t}^\star$ of $\mathrm{ILP}(w,\sigma)$ satisfies $\sum_it^\star_i\ge\sum_it_i\ge W^\star$, and its greedy allocation is complete, envy-free, and has welfare at least $\sum_it^\star_i\ge W^\star$. As the algorithm returns an allocation of maximum welfare among the allocations obtained for all guesses, the returned allocation has welfare at least $W^\star$.

For the running time, there are at most $n^{n(n-1)}$ candidate witness maps and, for each of them, at most $n^{n(n-1)}$ tightness maps, i.e., $2^{\mathcal{O}(n^2\log n)}$ guesses in total. For each guess, $\mathrm{ILP}(w,\sigma)$ is solved in $n^{\mathcal{O}(n)}\cdot\mathrm{poly}(n,m,\log V)$ time by Lemma~\ref{lem:witness-ilp}, and all other steps take polynomial time.

\paragraph{\bf A faster algorithm with a common good.}
Next, we show that if some good is demanded by all agents and every agent demands at least two goods, then guessing witnesses is unnecessary, which leads to an improved running time. Let $C=\{g_j\in G : d_{ij}>0 \text{ for every } i\in N\}$ be the set of goods demanded by all agents; we assume $|C|\ge1$ and $|S_i|\ge2$ for every agent $i$, which holds, in particular, if $|C|\ge2$. Our algorithm has two steps. First, it computes the best target utilities $\mathbf{t}$ by an ILP with $n$ variables, such that giving every agent $k$ exactly $L_{kj}=t_kd_{kj}$ copies of every good $g_j$ is envy-free. Second, it gives all remaining copies to two carefully chosen agents, which keeps the allocation envy-free.

For $\mathbf{t}\in\mathbb{Z}_{\ge0}^n$, we say that a partial allocation $\mathbf{Y}$ is \emph{non-wasteful} w.r.t.~$\mathbf{t}$ if $y_{kj}=L_{kj}=t_kd_{kj}$ for every agent $k$ and good $g_j$, i.e., every agent receives exactly what it needs for utility $t_k$. Let $r_j=q_j-\sum_{k\in N}L_{kj}$ be the number of remaining copies of $g_j$. For a pair $i\to k$, the goods of $S_i$ of which $k$ holds few enough copies in $\mathbf{Y}$ to be a witness for $i\to k$ are its \emph{possible witnesses}: $A_{ik}(\mathbf{t})=\bigl\{g_j\in S_i : t_k\,d_{kj}\le(t_i+1)\,d_{ij}-1\bigr\}$. Hence, agent $i$ does not envy agent $k$ in $\mathbf{Y}$ if and only if $A_{ik}(\mathbf{t})\neq\emptyset$.

\paragraph{Step 1: Computing the utilities by solving an ILP.}
For every pair $i\to k$, let $\rho_{ik}\in S_i$ be a good minimising $d_{kj}/d_{ij}$, i.e., the good of $S_i$ that limits the utility of $i$ for the bundle of $k$ in $\mathbf{Y}$. We compute $\mathbf{t}$ by the following ILP with an integer variable $t_i$ for every agent $i$.
\begin{alignat}{2}
  \max\quad & \textstyle\sum_{i\in N} t_i && \label{ilp:obj}\\
  \text{s.t.}\quad
  & t_k\, d_{k\rho_{ik}} \le (t_i+1)\, d_{i\rho_{ik}} - 1 &\quad& \forall\, i\to k, \label{ilp:ef}\\
  & \textstyle\sum_{k\in N} t_k\, d_{kj} \le q_j && \forall\, g_j\in G, \label{ilp:cap}\\
  & t_i \in \mathbb{Z}_{\ge0} && \forall\, i\in N. \label{ilp:int}
\end{alignat}
Constraint~\eqref{ilp:ef} states that $\rho_{ik}$ is a possible witness for $i\to k$, and constraint~\eqref{ilp:cap} states that there are enough copies of every good to give every agent its target utility.

\begin{lemma}{ilp-value}
\label{lem:ilp-value} 
A vector $\mathbf{t}\in\mathbb{Z}_{\ge0}^n$ is a solution of the ILP if and only if the partial allocation that is non-wasteful w.r.t.~$\mathbf{t}$ exists and is envy-free. Moreover, the ILP can be solved in $n^{\mathcal{O}(n)}\cdot\mathrm{poly}(n,m,\log V)$ time, and by a randomized algorithm in $(\log n)^{\mathcal{O}(n)}\cdot\mathrm{poly}(n,m,\log V)$ time.
\end{lemma}

Let $\mathbf{Y}$ be given by $y_{kj}=t_kd_{kj}$ for every agent $k$ and good $g_j$, so that $u_i(\mathbf{Y}_i)=t_i$ for every agent $i$. For a pair $i\to k$ and $\rho=\rho_{ik}$, the choice of $\rho_{ik}$ gives $u_i(\mathbf{Y}_k)=\min_{g_j\in S_i}\lfloor t_kd_{kj}/d_{ij}\rfloor=\lfloor t_kd_{k\rho}/d_{i\rho}\rfloor$. Hence, \eqref{ilp:ef} holds for $i\to k$ if and only if $\lfloor t_kd_{k\rho}/d_{i\rho}\rfloor\le t_i$, i.e., if and only if $i$ does not envy $k$ in $\mathbf{Y}$, and \eqref{ilp:cap} holds if and only if $\mathbf{Y}$ uses at most $q_j$ copies of every good $g_j$, i.e., is a partial allocation. For the running time, the goods $\rho_{ik}$ can be computed in $\mathcal{O}(n^2m)$ time. The ILP is feasible, as $\mathbf{t}=\mathbf{0}$ satisfies all constraints, and bounded by~\eqref{ilp:cap}. It has $n$ integer variables, no continuous variables, and $n(n-1)+m$ constraints, whose entries are bounded by $V$ in absolute value. Hence, both time bounds follow from Proposition~\ref{lem:lenstra}.

The ILP optimum is an upper bound on the welfare of every complete envy-free allocation.

\begin{lemma}{ilp-bound}
\label{lem:ilp-bound} 
Let $\mathbf{t}$ be an optimal solution of the ILP. Then every complete envy-free allocation has welfare at most $\sum_{i\in N}t_i$.
\end{lemma}

Let $\mathbf{X}^\star$ be a complete envy-free allocation, let $t^\star_i=u_i(\mathbf{X}^\star_i)$ for every agent $i$, and let $\mathbf{Y}^\star$ be non-wasteful w.r.t.~$\mathbf{t}^\star$. For every agent $i$ and good $g_j\in S_i$, we have $\lfloor x^\star_{ij}/d_{ij}\rfloor\ge t^\star_i$, so $x^\star_{ij}\ge t^\star_id_{ij}=y^\star_{ij}$, which also holds trivially for $g_j\notin S_i$. Thus, $\mathbf{Y}^\star$ is obtained from $\mathbf{X}^\star$ by removing copies, so it is a partial allocation, and since removing copies does not increase utilities, $u_i(\mathbf{Y}^\star_k)\le u_i(\mathbf{X}^\star_k)\le t^\star_i=u_i(\mathbf{Y}^\star_i)$ for every pair $i\to k$. Hence, $\mathbf{Y}^\star$ is envy-free, and $\mathbf{t}^\star$ is a solution of the ILP by Lemma~\ref{lem:ilp-value}. Thus, the welfare of $\mathbf{X}^\star$ is $\sum_it^\star_i\le\sum_it_i$.

\paragraph{Step 2: Completing the allocation.}

Let $\mathbf{Y}$ be an envy-free partial allocation that is non-wasteful w.r.t.~$\mathbf{t}$, e.g., the one given by Step~1. We give all remaining copies to two agents $k$ and $k'$. For a good $g_x\in C$, let $k(x)$ be an agent with the fewest copies of $g_x$ in $\mathbf{Y}$, i.e., minimising $t_id_{ix}$, breaking ties arbitrarily.
We choose a good $g_x\in C$, let $k=k(x)$, and give $k$ all $r_j$ remaining copies of every good $g_j\neq g_x$. Since $k$ receives no additional copies of $g_x$, and $g_x$ is a possible witness for every pair $i\to k$ by the choice of $k$, no agent envies $k$. We give all $r_x$ remaining copies of $g_x$ to an agent $k'\neq k$ such that every pair $i\to k'$ has a possible witness other than $g_x$, i.e., $A_{ik'}(\mathbf{t})\neq\{g_x\}$; then no agent envies $k'$ either. The following lemma shows that such $g_x$ and $k'$ always exist.

\begin{lemma}{completion}
\label{lem:completion}
Suppose that $|C|\ge1$ and $|S_i|\ge2$ for every agent $i$, and that $\mathbf{Y}$ is an envy-free partial allocation that is non-wasteful w.r.t.~$\mathbf{t}$. Then, there are a good $g_x\in C$ and an agent $k'\neq k(x)$ with $A_{ik'}(\mathbf{t})\neq\{g_x\}$ for every pair $i\to k'$, and Step~2 with this choice yields a complete envy-free allocation $\mathbf{X}$ with $u_i(\mathbf{X}_i)\ge t_i$ for every agent $i$. Moreover, $g_x$ and $k'$ can be found in polynomial time.
\end{lemma}

We first show that a suitable $g_x$ and $k'$ exist. We say that an agent $i$ \emph{blocks} an agent $k'\neq i$ for a good $g_x$ if $A_{ik'}(\mathbf{t})=\{g_x\}$; thus, we need $g_x\in C$ and $k'\neq k(x)$ such that no agent blocks $k'$ for $g_x$. We use the following simple observation.

\begin{observation}\label{obs:block}
If $i$ blocks $k'$ for $g_x$, then $t_{k'}d_{k'j}\ge(t_i+1)d_{ij}>t_id_{ij}$ for every good $g_j\in S_i\setminus\{g_x\}$, since $g_j$ is not a possible witness for $i\to k'$; i.e., $k'$ holds more copies of $g_j$ than $i$.
\end{observation}

Fix $g_a\in C$, let $z=k(a)$, and let $g_b\in S_z\setminus\{g_a\}$, which exists as $|S_z|\ge2$. Among the agents other than $z$, let $p$ hold the fewest copies of $g_b$, breaking ties by the fewest copies of the goods other than $g_a$ in total, and let $q$ hold the fewest copies of $g_a$.

First, we show that only $z$ can block $p$ for $g_a$. If some $i\neq z$ blocks $p$ for $g_a$, then by Observation~\ref{obs:block}, $i$ holds fewer copies than $p$ of every good in $S_i\setminus\{g_a\}\neq\emptyset$ and no copies of the goods outside $S_i$, so $i$ holds at most as many copies of $g_b$ as $p$ and fewer copies of the goods other than $g_a$ in total, contradicting the choice of $p$. Hence, if $z$ does not block $p$ for $g_a$, we choose $g_x=g_a$ and $k'=p\neq k(a)$.

Next, we show that if $z$ blocks $p$ for $g_a$, then $g_b\in C$, $k(b)=z$, and no agent blocks $q$ for $g_b$, so we can choose $g_x=g_b$ and $k'=q\neq k(b)$. By Observation~\ref{obs:block} applied to $g_b\in S_z$, every agent $i\neq z$ satisfies $t_id_{ib}\ge t_pd_{pb}>t_zd_{zb}$; thus, every agent demands $g_b$, i.e., $g_b\in C$, and $k(b)=z$. By Observation~\ref{obs:block} applied to $g_a$, which every agent demands, an agent blocking $q$ for $g_b$ holds fewer copies of $g_a$ than $q$, so it can only be $z$. However, $g_a\in A_{zp}(\mathbf{t})$ gives $t_qd_{qa}\le t_pd_{pa}<(t_z+1)d_{za}$, so $g_a\in A_{zq}(\mathbf{t})$, and $z$ does not block $q$ for $g_b$ either. Thus, in both cases, we obtain $g_x\in C$ and $k'\neq k(x)$ such that no agent blocks $k'$ for $g_x$, i.e., $A_{ik'}(\mathbf{t})\neq\{g_x\}$ for every pair $i\to k'$.

Finally, $\mathbf{X}$ is complete, and $u_i(\mathbf{X}_i)\ge t_i$ for every agent $i$, since we only add copies to $\mathbf{Y}$. For envy-freeness, it suffices to find, for every pair $i\to h$, a possible witness $g_j\in A_{ih}(\mathbf{t})$ of which $h$ received no additional copies, since then $x_{hj}\le(t_i+1)d_{ij}-1$ and thus $u_i(\mathbf{X}_h)\le t_i$. If $h\notin\{k,k'\}$, then $h$ received no additional copies, so we can take any $g_j\in A_{ih}(\mathbf{t})$, which is non-empty as $\mathbf{Y}$ is envy-free. If $h=k$, we take $g_j=g_x$, which $k$ did not receive and which lies in $A_{ik}(\mathbf{t})$, since $t_kd_{kx}\le t_id_{ix}<(t_i+1)d_{ix}$ by the choice of $k$. If $h=k'$, then $A_{ik'}(\mathbf{t})$ is non-empty, as $\mathbf{Y}$ is envy-free, and differs from $\{g_x\}$, as no agent blocks $k'$ for $g_x$ by the two cases above. Hence, $A_{ik'}(\mathbf{t})\setminus\{g_x\}$ is non-empty, and we take any good $g_j$ in it; agent $k'$ received no additional copies of $g_j$.
 Finally, computing all sets $A_{ik}(\mathbf{t})$ and all agents $k(x)$, and checking the choices of $g_x$ and $k'$, takes polynomial time.

\begin{theorem}{large-common}
\label{thm:large-common} 
\LEFCASMshort can be solved in $2^{\mathcal{O}(n\log n)}\cdot\mathrm{poly}(n,m,\log V)$ time, and in $2^{\mathcal{O}(n\log\log n)}\cdot\mathrm{poly}(n,m,\log V)$ time by a randomized algorithm, if some good is demanded by all agents and every agent demands at least two goods and in this case, a complete envy-free allocation always exists.
\end{theorem}

    For existence, the empty allocation is non-wasteful w.r.t.~$\mathbf{t}=\mathbf{0}$ and envy-free, as every agent has utility $0$ for every bundle, so Step~2 completes it to a complete envy-free allocation by Lemma~\ref{lem:completion}. For maximum welfare, let $\mathbf{t}$ be an optimal solution of the ILP in Step~1 with value $W_{ILP}=\sum_{i\in N}t_i$, and let $W^\star$ be the maximum welfare of a complete envy-free allocation. By Lemma~\ref{lem:ilp-value}, the partial allocation $\mathbf{Y}$ that is non-wasteful w.r.t.~$\mathbf{t}$ is envy-free, and by Lemma~\ref{lem:completion}, Step~2 completes it to a complete envy-free allocation $\mathbf{X}$ with welfare at least $W_{ILP}$, so $W^\star\ge W_{ILP}$. By Lemma~\ref{lem:ilp-bound}, $W^\star\le W_{ILP}$. Hence, $W_{ILP}=W^\star$, and $\mathbf{X}$ has maximum welfare. The running time follows from Lemmas~\ref{lem:ilp-value} and~\ref{lem:completion}, since $n^{\mathcal{O}(n)}=2^{\mathcal{O}(n\log n)}$ and $(\log n)^{\mathcal{O}(n)}=2^{\mathcal{O}(n\log\log n)}$.

\subsection{FPT wrt Total Number of Copies of Goods}

By Theorem~\ref{thm:hard-m} and Theorem~\ref{thm:np-hard-welfare}, we cannot hope for an FPT algorithm wrt $m$ or $q_{max}$. Thus, we consider the parameter $m+q_{max}$, and design an FPT algorithm. We design an algorithm wrt even the total number of copies of goods $K=\sum_{g_j\in G}q_j$ which is at most $mq_{max}$.

Since every non-empty bundle contains at least one copy of a good, at most $k=\min\{n,K\}$ agents receive a non-empty bundle in any allocation. Our algorithm therefore guesses the vectors that will form the bundles, without deciding who gets which, and then assigns them to the agents by finding a perfect matching in a bipartite graph, if it exists.

\begin{sloppypar}
    {\bf Algorithm.} We enumerate all tuples $\mathcal{P}=(\mathbf{P}_1,\ldots,\mathbf{P}_n)$, where $\mathbf{P}_r=(p_{r1},\ldots,p_{rm})\in\mathbb{Z}_{\ge0}^m$ and $p_{rj}$ denote the number of copies of $g_j$ in $\mathbf{P}_r$, such that $\mathbf{P}_r=\mathbf{0}$ for every $r>k$ and $\sum_{r=1}^{n}p_{rj}=q_j$ for every good $g_j$. 

For an agent $i$, let $u_i^\star(\mathcal{P})=\max_{r\in[n]}u_i(\mathbf{P}_r)$ be its largest utility for a vector in $\mathcal{P}$, and let $W(\mathcal{P})=\sum_{i\in N}u_i^\star(\mathcal{P})$.  
\end{sloppypar}

For every tuple $\mathcal{P}$, we construct a bipartite graph $H_{\mathcal{P}}$ on the vertex set $N\uplus[n]$, with an edge $\{i,r\}$ iff $u_i(\mathbf{P}_r)=u_i^\star(\mathcal{P})$, i.e., iff agent $i$ does not prefer any vector of $\mathcal{P}$ to $\mathbf{P}_r$. If there exists a perfect matching $M_{\mathcal{P}}$ in the graph $H_{\mathcal{P}}$, we create the allocation $\mathbf{X}^{\mathcal{P}}$ as follows: for every edge $\{i,r\}\in M_{\mathcal{P}}$, we allocate $p_{rj}$ copies of good $g_j$ to agent $i$, for every good $g_j$, i.e., $\mathbf{X}^{\mathcal{P}}_i=\mathbf{P}_r$. If $H_{\mathcal{P}}$ does not have a perfect matching for any tuple ${\mathcal{P}}$, we return that no EF allocation exists. Otherwise, we return the allocation $\mathbf{X}^{\mathcal{P}}$ for a tuple $\mathcal{P}$ with maximum $W(\mathcal{P})$ among all tuples $\mathcal{P}$ for which $H_{\mathcal{P}}$ has a perfect matching.

\begin{theorem}{split-welfare}
\label{thm:split-welfare}
\LEFCASMshort can be solved in $2^{\mathcal{O}(K\log K)}\cdot\mathrm{poly}(n,m)$ time, where $K=\sum_{g_j\in G}q_j$ is the total number of copies.
\end{theorem}

\begin{proof}
We first show that, for every tuple $\mathcal{P}$ such that $H_{\mathcal{P}}$ has a perfect matching $M_{\mathcal{P}}$, the allocation $\mathbf{X}^{\mathcal{P}}$ is complete and envy-free and has welfare $W(\mathcal{P})$. Since $M_{\mathcal{P}}$ is a perfect matching, every vector of $\mathcal{P}$ is given to exactly one agent, so every good $g_j$ is assigned $\sum_{r=1}^np_{rj}=q_j$ times, and $\mathbf{X}^{\mathcal{P}}$ is complete. Moreover, the bundles of $\mathbf{X}^{\mathcal{P}}$ are exactly the vectors of $\mathcal{P}$, and every agent $i$ receives a vector $\mathbf{P}_r$ with $\{i,r\}\in H_{\mathcal{P}}$, i.e., with $u_i(\mathbf{P}_r)=u_i^\star(\mathcal{P})$. Hence, agent $i$ values its own bundle at least as much as every other bundle, so $\mathbf{X}^{\mathcal{P}}$ is envy-free, and its welfare is $\sum_{i\in N}u_i^\star(\mathcal{P})=W(\mathcal{P})$.

Now suppose that a complete envy-free allocation exists, and let $\mathbf{Y}$ be one with maximum welfare. Order the agents as $i_1,\ldots,i_n$ so that the at most $k$ agents with a non-empty bundle come first, and let $\mathbf{P}_r=\mathbf{Y}_{i_r}$ for every $r\in[n]$. Since $\mathbf{Y}$ is complete, $\mathcal{P}=(\mathbf{P}_1,\ldots,\mathbf{P}_n)$ is a tuple considered by the algorithm. Since $\mathbf{Y}$ is envy-free, $u_{i_r}(\mathbf{P}_r)=u_{i_r}(\mathbf{Y}_{i_r})\ge u_{i_r}(\mathbf{Y}_{i'})$ for every agent $i'$, i.e., $u_{i_r}(\mathbf{P}_r)=u_{i_r}^\star(\mathcal{P})$. Hence, the edges $\{i_r,r\}$ for $r\in[n]$ form a perfect matching of $H_{\mathcal{P}}$, and the welfare of $\mathbf{Y}$ is $W(\mathcal{P})$. Since the algorithm returns $\mathbf{X}^{\mathcal{P}'}$ for a tuple $\mathcal{P}'$ with maximum $W(\mathcal{P}')$ among all tuples whose graph has a perfect matching, it returns a complete envy-free allocation with welfare $W(\mathcal{P}')\ge W(\mathcal{P})$, which is maximum by the choice of $\mathbf{Y}$. Conversely, if a complete envy-free allocation exists, then, as argued above, some $H_{\mathcal P}$ has a perfect matching; hence, if the algorithm returns that none exists, there is none.

For the running time, a tuple is determined by placing each of the $K$ copies into one of the first $k$ vectors, so there are at most $k^K\le K^K=2^{K\log K}$ tuples. For each tuple, computing all values $u_i^\star(\mathcal{P})$ takes $\mathcal{O}(nkm)$ time, as only $k$ vectors are non-zero, and a perfect matching in $H_{\mathcal{P}}$ can be found in $\mathcal{O}(n^{2.5})$ time~\cite{HopcroftK73}. Thus, the total running time is $2^{\mathcal{O}(K\log K)}\cdot\mathrm{poly}(n,m)$.
\end{proof}

\begin{corollary}
\label{cor:split-welfare-mq}
\LEFCASMshort is FPT parameterized by $m+q_{\max}$.
\end{corollary}

\paragraph{\bf Few demand types.}
The running time of Theorem~\ref{thm:split-welfare} does not depend on the demands of the agents. We now show that if the agents have few distinct demand vectors, then the dependence on $K$ improves from $2^{\mathcal{O}(K\log K)}$ to $2^{\mathcal{O}(\tau+K)}$, where $\tau\le n$ is the number of distinct demand vectors; this is faster whenever $\tau<K\log K$. We say that agents with the same demand vector have the same \emph{type}. Type $\ell\in[\tau]$ has demand vector $\mathbf{d}^{(\ell)}=(d^{(\ell)}_1,\ldots,d^{(\ell)}_m)\in\mathbb{Z}_{\ge0}^m$ and $n_\ell>0$ agents; we let $S_\ell=\{g_j\in G: d^{(\ell)}_j>0\}$, and $u_\ell$ be the common utility function of the agents of type $\ell$.
 Since agents of the same type have identical preferences, they all have the same utility in every envy-free allocation. Our algorithm therefore guesses one utility per type, and checks with a dynamic program (DP) whether these utilities can be extended to a complete envy-free allocation.

\begin{lemma}{utility-sum}
\label{lem:utility-sum}
If $\mathbf{X}$ is envy-free and every agent of type $\ell$ has utility $t_\ell$, then $\sum_{\ell=1}^{\tau}n_\ell t_\ell\le K$.
\end{lemma}

Consider an agent $i$ with $u_i(\mathbf{X}_i)=t\ge1$. Then $S_i\neq\emptyset$, and for every good $g_j\in S_i$, agent $i$ receives at least $t\,d_{ij}\ge t$ copies of $g_j$. Hence, $\mathbf{X}_i$ contains at least $u_i(\mathbf{X}_i)$ copies, which trivially also holds if $u_i(\mathbf{X}_i)=0$. Since the bundles are disjoint and there are $K$ copies in total, $\sum_{i\in N}u_i(\mathbf{X}_i)\le K$, and $\sum_{i\in N}u_i(\mathbf{X}_i)=\sum_{\ell=1}^{\tau}n_\ell t_\ell$ by assumption.

We call $\mathbf{t}=(t_1,\ldots,t_\tau)\in\mathbb{Z}_{\ge0}^{\tau}$ a vector of \emph{target utilities}, where $t_\ell$ is the utility we aim for every agent of type $\ell$. By Lemma~\ref{lem:utility-sum}, it suffices to consider target utilities with $\sum_{\ell=1}^{\tau}n_\ell t_\ell\le K$, which we call \emph{bounded}.

We call $\mathbf{t}$ \emph{EF-extendible} if there is a complete envy-free allocation in which every agent of type $\ell$ has utility exactly $t_\ell$, for every $\ell\in[\tau]$.

{\bf Algorithm.} We enumerate all bounded target utilities $\mathbf{t}$, and for each of them, we decide with the DP below whether $\mathbf{t}$ is EF-extendible, and if so, compute a corresponding allocation. Among all bounded EF-extendible target utilities, we return the allocation for one with maximum welfare $\sum_\ell n_\ell t_\ell$. If no bounded vector of target utilities is EF-extendible, we return that no EF allocation exists.

\begin{sloppypar}
{\bf DP for a fixed $\mathbf{t}$.} We order the agents arbitrarily as $1,\ldots,n$, and let $\ell(i)$ be the type of agent $i$. A \emph{supply vector} is a vector $\mathbf{q}'=(q'_1,\ldots,q'_m)$ with $0\le q'_j\le q_j$ for every good $g_j$; intuitively, $q'_j$ is the number of remaining copies of $g_j$ after a partial allocation. For $0\le i\le n$ and a supply vector $\mathbf{q}'$, the entry $T[i,\mathbf{q}']\in\{\mathtt{true},\mathtt{false}\}$ indicates whether there are bundles $\mathbf{Z}_1,\ldots,\mathbf{Z}_i$ for agents $1,\ldots,i$ such that
\end{sloppypar}
\begin{enumerate}[label=(\roman*),leftmargin=*]
  \item exactly $q'_j$ copies of every good $g_j$ remain, i.e., $\sum_{i'=1}^{i}z_{i'j}=q_j-q'_j$;
  \item every agent $i'\le i$ has utility exactly $t_{\ell(i')}$, i.e., $u_{\ell(i')}(\mathbf{Z}_{i'})=t_{\ell(i')}$; and
  \item no agent of any type $\tilde{\ell}\in[\tau]$ values $\mathbf{Z}_{i'}$, for $i'\le i$, more than $t_{\tilde{\ell}}$, i.e., $u_{\tilde{\ell}}(\mathbf{Z}_{i'})\le t_{\tilde{\ell}}$.
\end{enumerate}
We set $T[0,\mathbf{q}]=\mathtt{true}$ and $T[0,\mathbf{q}']=\mathtt{false}$ for every supply vector $\mathbf{q}'\neq\mathbf{q}$. For $i\ge1$ and $\ell=\ell(i)$, we call a bundle $\mathbf{Z}=(z_1,\ldots,z_m)$ \emph{valid} for agent $i$ and supply vector $\mathbf{q}'$ if
(a) $0\le z_j\le q_j-q'_j$ for every good $g_j$;
(b) $u_\ell(\mathbf{Z})\ge t_\ell$, i.e., $z_j\ge t_\ell d^{(\ell)}_j$ for every $g_j\in S_\ell$; and
(c) $u_{\tilde{\ell}}(\mathbf{Z})\le t_{\tilde{\ell}}$ for every type $\tilde{\ell}\in[\tau]$, i.e., some $g_j\in S_{\tilde{\ell}}$ has $z_j<(t_{\tilde{\ell}}+1)d^{(\tilde{\ell})}_j$.
Note that (b) and (c) for $\tilde{\ell}=\ell$ imply $u_\ell(\mathbf{Z})=t_\ell$. Then
\[
  T[i,\mathbf{q}']\;=\;\bigvee_{\mathbf{Z}\text{ valid for } i \text{ and } \mathbf{q}'}T[i-1,\mathbf{q}'+\mathbf{Z}].
\]

\begin{lemma}{dp-fixed-target}
\label{lem:dp-fixed-target} 
For fixed bounded target utilities $\mathbf{t}$, the DP decides in $\mathcal{O}(n\cdot 3^K\cdot\tau m)$ time whether $\mathbf{t}$ is EF-extendible, and if so, produces a complete envy-free allocation in which every agent of type $\ell$ has utility exactly $t_\ell$.
\end{lemma}

We first show by induction on $i$ that $T[i,\mathbf{q}']=\mathtt{true}$ if and only if there are bundles $\mathbf{Z}_1,\ldots,\mathbf{Z}_i$ satisfying (i)--(iii) for $\mathbf{q}'$. For $i=0$, conditions (ii) and (iii) hold vacuously, and (i) holds if and only if $\mathbf{q}'=\mathbf{q}$, as set in the DP. Now let $i\ge1$, and assume that the statement holds for $i-1$ and every supply vector.

Suppose that $T[i,\mathbf{q}']=\mathtt{true}$. By the recurrence, there is a bundle $\mathbf{Z}$ that is valid for $i$ and $\mathbf{q}'$ with $T[i-1,\mathbf{q}'+\mathbf{Z}]=\mathtt{true}$; by (a), $\mathbf{q}'+\mathbf{Z}$ is a supply vector. By the induction hypothesis, there are bundles $\mathbf{Z}_1,\ldots,\mathbf{Z}_{i-1}$ satisfying (i)--(iii) for $\mathbf{q}'+\mathbf{Z}$. Let $\mathbf{Z}_i=\mathbf{Z}$. Then (i) holds for $\mathbf{q}'$, as $\mathbf{Z}_i$ uses $\mathbf{Z}$ of the remaining supply $\mathbf{q}'+\mathbf{Z}$. Condition (ii) holds for $i'<i$ by the induction hypothesis, and for $i$ as (b) and (c) give $u_{\ell(i)}(\mathbf{Z}_i)=t_{\ell(i)}$. Condition (iii) holds for $i'<i$ by the induction hypothesis, and for $i$ by (c).

Conversely, suppose that bundles $\mathbf{Z}_1,\ldots,\mathbf{Z}_i$ satisfy (i)--(iii) for $\mathbf{q}'$, and let $\mathbf{Z}=\mathbf{Z}_i$. Then $\mathbf{Z}_1,\ldots,\mathbf{Z}_{i-1}$ satisfy (i)--(iii) for $\mathbf{q}'+\mathbf{Z}$, which is a supply vector as $q'_j+z_j=q_j-\sum_{i'<i}z_{i'j}\le q_j$. Hence, $T[i-1,\mathbf{q}'+\mathbf{Z}]=\mathtt{true}$ by the induction hypothesis. Moreover, $\mathbf{Z}$ is valid for $i$ and $\mathbf{q}'$: (a) holds as $0\le z_j\le q_j-q'_j$, (b) holds by (ii) for agent $i$, and (c) holds by (iii) for $\mathbf{Z}_i$. Thus, $T[i,\mathbf{q}']=\mathtt{true}$ by the recurrence.

We now show that $T[n,\mathbf{0}]=\mathtt{true}$ if and only if $\mathbf{t}$ is EF-extendible. Suppose that $T[n,\mathbf{0}]=\mathtt{true}$, and let $\mathbf{Z}_1,\ldots,\mathbf{Z}_n$ satisfy (i)--(iii) for $\mathbf{0}$. Giving $\mathbf{Z}_i$ to agent $i$ yields a complete allocation by (i), in which every agent of type $\ell$ has utility exactly $t_\ell$ by (ii). It is envy-free: for agents $i,k$, condition (iii) for $\mathbf{Z}_k$ and type $\ell(i)$ gives $u_i(\mathbf{Z}_k)=u_{\ell(i)}(\mathbf{Z}_k)\le t_{\ell(i)}=u_i(\mathbf{Z}_i)$. Conversely, let $\mathbf{X}$ be a complete envy-free allocation in which every agent of type $\ell$ has utility exactly $t_\ell$. Then $\mathbf{X}_1,\ldots,\mathbf{X}_n$ satisfy (i) for $\mathbf{0}$ and (ii). For (iii), consider a bundle $\mathbf{X}_{i'}$ and a type $\tilde{\ell}$. Since $n_{\tilde{\ell}}>0$, there is an agent $k$ of type $\tilde{\ell}$, and since $k$ does not envy $i'$, $u_{\tilde{\ell}}(\mathbf{X}_{i'})\le u_{\tilde{\ell}}(\mathbf{X}_k)=t_{\tilde{\ell}}$. Hence, $T[n,\mathbf{0}]=\mathtt{true}$. In this case, a corresponding allocation can be recovered by standard backtracking, in time proportional to the size of the table.

For the running time, the number of supply vectors is $\prod_{g_j\in G}(q_j+1)\le\prod_{g_j\in G}2^{q_j}=2^K$, so the table has at most $(n+1)\cdot2^K$ entries. To fill all entries for a fixed $i$, we iterate over all pairs of a supply vector $\mathbf{q}'$ and a bundle $\mathbf{Z}$ with $\mathbf{q}'+\mathbf{Z}\le\mathbf{q}$. For a single good $g_j$, the number of pairs $(q'_j,z_j)$ of non-negative integers with $q'_j+z_j\le q_j$ is $\binom{q_j+2}{2}\le3^{q_j}$, so there are at most $\prod_{g_j\in G}3^{q_j}=3^K$ such pairs. For each pair, we check whether $\mathbf{Z}$ is valid in $\mathcal{O}(\tau m)$ time. Thus, the DP takes $\mathcal{O}(n\cdot3^K\cdot\tau m)$ time.

\begin{theorem}{dp-welfare}
\label{thm:dp-welfare}
\LEFCASMshort can be solved in $2^{\mathcal{O}(\tau+K)}\cdot\mathrm{poly}(n,m)$ time, where $\tau$ is the number of distinct demand vectors and $K$ is the total number of copies.
\end{theorem}

\begin{proof}
Let $\mathbf{X}$ be a complete envy-free allocation. Since all agents of type $\ell$ have identical preferences, they all have the same utility $t_\ell$ in $\mathbf{X}$. By Lemma~\ref{lem:utility-sum}, $\sum_\ell n_\ell t_\ell\le K$. Thus, $\mathbf{t}=(t_1,\ldots,t_\tau)$ is a bounded and EF-extendible vector of target utilities, and the welfare of $\mathbf{X}$ is $\sum_\ell n_\ell t_\ell$. Hence, it suffices to consider bounded target utilities.

By Lemma~\ref{lem:dp-fixed-target}, for all bounded target utilities $\mathbf{t}$, the algorithm either obtains a complete envy-free allocation in which every agent of type $\ell$ has utility exactly $t_\ell$, or correctly decides that no such allocation exists. Taking $\mathbf{X}$ above with maximum welfare, the algorithm returns a complete envy-free allocation with welfare at least that of $\mathbf{X}$, i.e., with maximum welfare. If no bounded vector of target utilities is EF-extendible, then, by the above, no complete envy-free allocation exists, and the algorithm correctly returns that no EF allocation exists.

For the running time, the number of bounded vectors of target utilities is at most the number of vectors $\mathbf{t}\in\mathbb{Z}_{\ge0}^{\tau}$ with $\sum_\ell t_\ell\le K$, which is $\binom{\tau+K}{K}\le2^{\tau+K}$. By Lemma~\ref{lem:dp-fixed-target}, the DP takes $\mathcal{O}(n\cdot3^K\cdot\tau m)$ time for each of them. Thus, the total running time is $2^{\tau+K}\cdot\mathcal{O}(n\cdot3^K\cdot\tau m)=2^{\mathcal{O}(\tau+K)}\cdot\mathrm{poly}(n,m)$.
\end{proof}

\section{Envy-Freeness up to One Task and Pareto Optimality}
\label{sec:ef1-po}

As envy-freeness and Pareto optimality are incompatible (Section~\ref{sec:welfare}), we relax envy-freeness to envy-freeness up to one \emph{task}. An allocation $\mathbf{X}$ is \emph{EF1} if $u_i(\mathbf{X}_i)\ge u_i(\mathbf{X}_k-\mathbf{d}_i)$ for all agents $i,k$, where $\mathbf{X}_k-\mathbf{d}_i$ is obtained from $\mathbf{X}_k$ by removing one task of $i$ (if $\mathbf{X}_k$ contains no task of $i$, then $u_i(\mathbf{X}_k)=0$). Equivalently, $u_i(\mathbf{X}_k)\le u_i(\mathbf{X}_i)+1$ for all agents $i,k$.

\begin{observation}
\label{obs:ef1-equal}
A complete EF1 allocation always exists and can be computed in $\mathcal{O}(nm)$ time.
\end{observation}
\begin{proof}
Give every agent $\lfloor q_j/n\rfloor$ or $\lceil q_j/n\rceil$ copies of every good $g_j$. Then $x_{kj}\le x_{ij}+1$ for all agents $i,k$, so $u_i(\mathbf{X}_k)\le\min_{g_j\in S_i}\lfloor(x_{ij}+1)/d_{ij}\rfloor\le u_i(\mathbf{X}_i)+1$.
\end{proof}

Such an allocation may be far from Pareto optimal, e.g., for two agents demanding one copy of each of two goods with one copy each, it has welfare $0$, whereas giving both goods to one agent has welfare $1$ and is EF1.

\begin{proposition}[{\cite[Theorem~5.4]{10.1145/2739040}}]
\label{prop:smm}
\textsc{SequentialMinMax} computes in $\mathcal{O}(n^2m\cdot q_{\max})$ time a partial allocation $\mathbf{Y}$ that is EF1 and Pareto optimal with respect to all (partial or complete) allocations, and non-wasteful w.r.t.~$\mathbf{t}$, where $t_i=u_i(\mathbf{Y}_i)$.
\end{proposition}

Let $r_j=q_j-\sum_{k\in N}t_kd_{kj}$ be the number of remaining copies of $g_j$.

\begin{observation}
\label{obs:po-preserved}
Every agent $i$ has a good $g_j\in S_i$ with $r_j<d_{ij}$, as otherwise giving $i$ one more task would contradict the Pareto optimality of $\mathbf{Y}$.
\end{observation}

For agents $i\neq k$, we say that $i$ \emph{objects} to $k$ if $t_kd_{kj}+r_j\ge(t_i+2)d_{ij}$ for every $g_j\in S_i$, i.e., if $i$ would derive at least $t_i+2$ tasks from the bundle of $k$ after $k$ receives all remaining copies. We give all remaining copies to an agent to which no agent objects.

\begin{theorem}{ef1-po}
\label{thm:ef1-po}
For every instance $(N,G,Q,D)$, a complete allocation that is EF1 and Pareto optimal always exists and can be computed in $\mathcal{O}(n^2m\cdot q_{\max})$ time.
\end{theorem}
\begin{proof}
Let $\mathbf{Y}$ be the partial allocation of Proposition~\ref{prop:smm}. We first show that some agent has no objector. Suppose for contradiction that every agent has an objector. Then there is a cycle $a_0,a_1,\ldots,a_s=a_0$ with $s\ge2$ in which $a_p$ objects to $a_{p+1}$ for every $p<s$; let $t_{a_0}\ge t_{a_p}$ for every $p$. By Observation~\ref{obs:po-preserved}, some good $g_j\in S_{a_0}$ has $r_j<d_{a_0j}$. Let $B=(t_{a_0}+1)d_{a_0j}$.

\begin{claim}\label{clm:many-copies}
For every $p\in [s]$, $t_{a_p}d_{a_pj}>B$.
\end{claim}
\begin{proof}
We use induction on $p$. For $p=1$, as $r_j<d_{a_0j}$ and $a_0$ objects to $a_1$, we have $t_{a_1}d_{a_1j}\ge(t_{a_0}+2)d_{a_0j}-r_j>B$. Now assume that $t_{a_p}d_{a_pj}>B$ for some $1\le p<s$. Then $t_{a_p}\ge1$, and as $t_{a_p}\le t_{a_0}$, we get $d_{a_pj}>B/t_{a_p}\ge B/t_{a_0}>d_{a_0j}>r_j$. Hence, $g_j\in S_{a_p}$, and as $a_p$ objects to $a_{p+1}$, we have $t_{a_{p+1}}d_{a_{p+1}j}\ge(t_{a_p}+2)d_{a_pj}-r_j>(t_{a_p}+1)d_{a_pj}>B$.
\end{proof}

For $p=s$, Claim~\ref{clm:many-copies} gives $t_{a_0}d_{a_0j}>(t_{a_0}+1)d_{a_0j}$, a contradiction. Hence, some agent $k$ has no objector. Let $\mathbf{X}$ be obtained by giving all remaining copies to $k$; then $\mathbf{X}$ is complete, and $u_i(\mathbf{X}_i)\ge t_i$ for every agent $i$. For $\ell\neq k$, $\mathbf{X}_\ell=\mathbf{Y}_\ell$, so $u_i(\mathbf{X}_\ell)\le t_i+1$ for every agent $i$, as $\mathbf{Y}$ is EF1. For $i\neq k$, since $i$ does not object to $k$, some $g_j\in S_i$ has $x_{kj}=t_kd_{kj}+r_j<(t_i+2)d_{ij}$, so $u_i(\mathbf{X}_k)\le t_i+1$. Hence, $\mathbf{X}$ is EF1, and it is Pareto optimal, as any allocation that Pareto dominates $\mathbf{X}$ also Pareto dominates $\mathbf{Y}$. Computing $\mathbf{Y}$ takes $\mathcal{O}(n^2m\cdot q_{\max})$ time, and checking all objections takes $\mathcal{O}(n^2m)$ time.
\end{proof}

\section{Outlook}
We initiated the study of envy-free complete allocations under Leontief preferences, which model resources such as GPUs, cores and memory that a web platform shares among its services. Unlike for additive valuations, envy-free allocations always exist if every agent demands at least two goods, and their existence is decidable efficiently if every agent demands at most one good. Maximizing welfare among them is hard even when existence is guaranteed, but tractable for few agents or few copies, and EF1 and Pareto optimality are compatible. The main open question is whether existence is polynomial-time decidable when some agents value at most one good positively and the others value at least two goods positively, as when single-resource jobs share a cluster with multi-resource services.

\bibliographystyle{plainnat}
\bibliography{references}

\end{document}